\documentclass[onefignum,onetabnum]{siamart250211}
\usepackage{geometry}
\usepackage{amsmath,amssymb,amsfonts}
\usepackage{graphicx}
\usepackage{epstopdf}

\newsiamremark{remark}{Remark}
\newsiamremark{hypothesis}{Hypothesis}
\crefname{hypothesis}{Hypothesis}{Hypotheses}
\newsiamthm{claim}{Claim}

\usepackage{bm}          
\usepackage{mathrsfs}    
\usepackage{mathtools}   
\usepackage{physics}
\usepackage{xcolor}
\usepackage{enumitem}
\usepackage{booktabs}
\usepackage{threeparttable}
\usepackage{algorithm} 
\usepackage{algpseudocode}  

\ifpdf
  \DeclareGraphicsExtensions{.eps,.pdf,.png,.jpg}
\else
  \DeclareGraphicsExtensions{.eps}
\fi

\headers{Survival and invasion dynamics in cell populations}{S. M. C. Abo and R. E. Baker}

\title{Survival and invasion dynamics in cell populations: an analytical framework for threshold behaviour in nonlinear age-structured models}

\author{St\'ephanie M. C. Abo\thanks{Mathematical Institute, University of Oxford, UK 
  (\email{stephanie.abo@maths.ox.ac.uk}).}
\and Ruth E. Baker\footnotemark[1]}

\newcommand{\D}{\mathrm{d}}

\newcommand{\FKPP}{Fisher--KPP}

\usepackage{amsopn}

\setlist[itemize,1]{label={\tiny\raisebox{0.8ex}{\textbullet}}, itemsep=-2pt, topsep=4pt}
\setlist[enumerate,1]{itemsep=-2pt, topsep=4pt}

\definecolor{stephaniecolor}{RGB}{144, 12, 63}

\usepackage{pdfpages}   

\hypersetup{
  pdftitle={Survival and invasion dynamics in cell populations: an analytical framework for threshold behaviour in nonlinear age-structured models},
  pdfauthor={St\'ephanie M. C. Abo and Ruth E. Baker},
  pdfkeywords={age structure, cell division, population dynamics, travelling waves, \FKPP}
}

\headers{Invasion dynamics in cell populations}{S. Abo and R. E. Baker}

\begin{document}

\maketitle

\begin{abstract}
Cell populations invade through a combination of proliferation and motility. Proliferation depends on the internal timing of cell division: how long cells take to complete the cell cycle. This timing varies substantially within (and across) cell types, creating age structure where cells at different times since their last division have different propensities to divide. Classical mathematical models of cell spreading treat division as memoryless and predict exponential cell-cycle-time distributions. Lineage tracing, by contrast, reveals peaked, gamma-like distributions that indicate a maturation delay leading to a fertility window. This gap motivates a modelling framework that incorporates age-dependent cell division rates while retaining analytical tractability. We address this through a moment-hierarchy framework that tracks time since cell division, with age resetting to zero at division. The framework yields explicit formulae for steady-state age distributions, cell-cycle-time distributions, and invasion speeds. For age-independent rates, we recover classical \FKPP{}. Three fundamental principles emerge. First, age structure systematically reduces a population's carrying capacity and narrows the viable parameter range for positive steady states. Second, classical linear theory overestimates invasion speeds; the true minimal speed is slower when division is age-dependent. Third, the parameter condition for population survival is identical to the condition for a positive invasion speed.
\end{abstract}



\section{Introduction}\label{sec:intro}
The coordinated behaviour of motile and proliferative cells is fundamental to embryonic development, tissue regeneration, wound healing, and tumour invasion. In such processes, spatial expansion occurs through \emph{collective cell invasion}, where cells move as cohesive groups with mechanical and biochemical coupling \cite{aoki_2017}. These dynamics depend on both external cues (adhesion, confinement) and internal cell-cycle state \cite{heinrich_2020, donker_2022, falco_2023}. The cell cycle imposes a necessary maturation delay before division can occur, acting as a key biological constraint. The progression time through the cell cycle, measured as the intermitotic time (IMT), exhibits marked heterogeneity in experimental data \cite{smith_1973, sakaue_2008}. High-resolution lineage tracing reveals IMT distributions that are typically gamma-like, with defined peaks and right skew, rather than exponential as assumed in memoryless birth models \cite{yates_2017} (Figure~\ref{fig:intro-diagram}(a)). This discrepancy indicates that models assuming constant division rates across the cell cycle do not account for essential biological regulation.

\begin{figure}[htb]
  \centering
  \includegraphics[width=\textwidth]{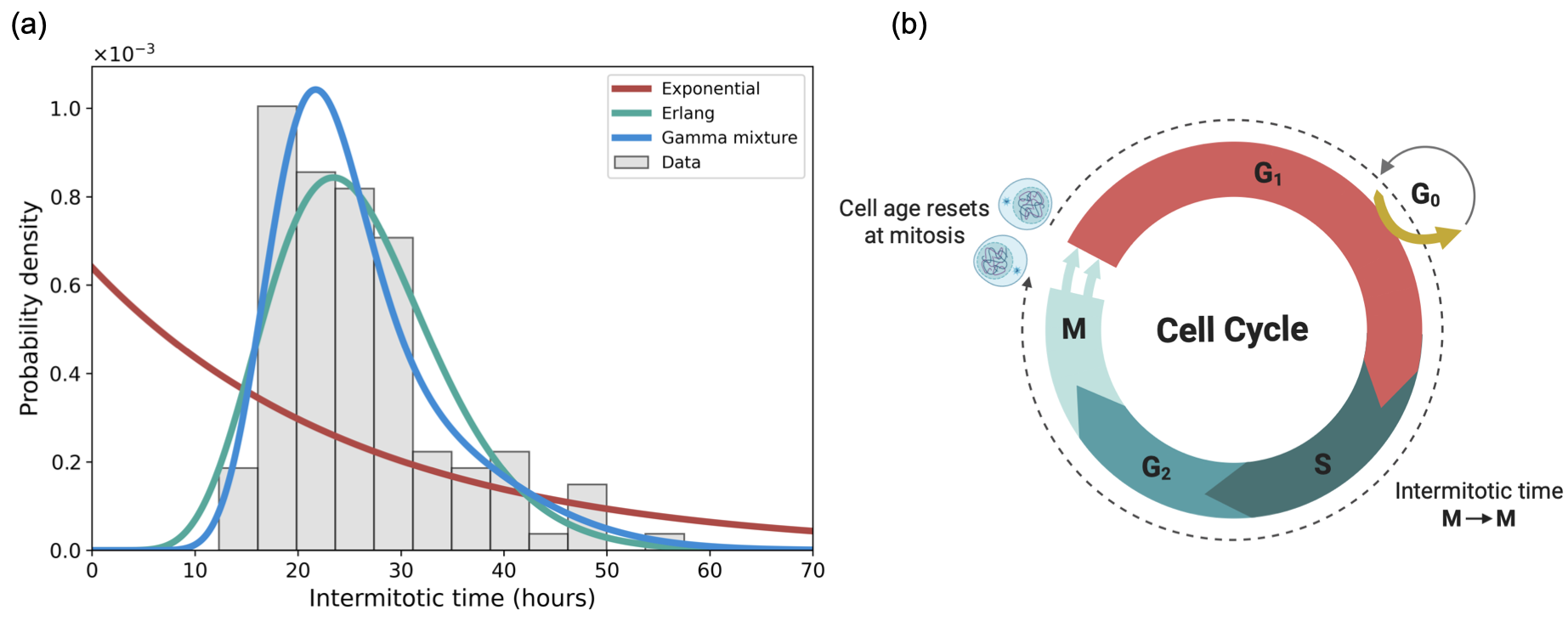}
  \caption{(a) Distribution of intermitotic times (IMTs) from lineage-tracing data (grey bars) with fitted exponential (red), Erlang (green), and gamma mixture (blue) models. (b) Schematic of the cell cycle. Cells progress through four main phases: Gap 1 (G1, growth and preparation for DNA synthesis), Synthesis (S, DNA replication), Gap 2 (G2, preparation for mitosis), and Mitosis (M, cell division; here mitosis and cell division are used interchangeably). In our model, cell age (defined as time since last division) resets to zero upon mitosis, providing a coarse-grained representation of cell-cycle progression. The intermitotic time (IMT) is measured from one mitosis to the next (M\(\to\)M). Experimental IMT data reproduced with permission, courtesy of Richard L. Mort (Lancaster University) and Matthew J. Ford (University of Cambridge), provided by Christian A. Yates (University of Bath). Cell-cycle stages adapted from \cite{jiang_2025} (BioRender).}
  \label{fig:intro-diagram}
\end{figure}

Non-exponential division timing has motivated structured modelling approaches that account for how internal states constrain and delay cell division (Figure~\ref{fig:intro-diagram}(b)). These approaches introduce continuous variables representing internal timing \cite{jones_2018, kynaston_2022, ocal_2025, kang_2022}. \emph{Cell age} can denote chronological age, phase duration, or generational time; age-structured partial differential equations (PDEs) provide a general formalism accommodating this variability. Here, we focus on cell age as time since last division---a biologically grounded assumption aligning with experimental measurement of cell-cycle times. This simplification aggregates sub-stages while providing a meaningful foundation for future extensions incorporating finer-grained structure.

Classical mathematical models of population spread, such as the \FKPP{} (Fisher--Kolmogorov--Petrovsky--Piskunov) equation, assume cells divide at a constant rate independent of time since division. These models do not account for the finite time required to complete the cell cycle, which affects division-time distributions and the capacity of populations to persist or invade. Age-structured models address this gap by tracking the joint evolution of density and internal time. The foundational works of McKendrick \cite{dietz_1997} and von~Foerster \cite{vonfoerster_1959} formalised transport equations where individuals progress deterministically through age. In the simplest linear setting, populations grow exponentially and approach stable age distributions. Gurtin and MacCamy \cite{gurtin_1979, gurtin_1981} extended this to include density dependence, spatial diffusion, and biologically motivated division kernels where division rates peak at intermediate ages. Their formulation combined division and death within a single loss term, with age dependence appearing only in the boundary condition. This simplified the mathematics but limited the capacity to distinguish how ageing, division, and death separately influence dynamics.

Spatial theory advanced with the work of Al-Omari and Gourley \cite{alomari_2002}, who established the existence of travelling waves in age-structured reaction-diffusion models. Gourley \cite{gourley_2005} proved linear stability of travelling fronts, and Li, Mei, and Wong \cite{li_2008} extended this to nonlinear stability, though initially restricted to small maturation delays. Mei and Wong \cite{mei_2009} later proved nonlinear stability for arbitrarily large delays. In applied contexts, structured PDEs have been used to model cell-cycle regulation in proliferating tissues \cite{nixson_2025} and to connect multi-stage stochastic frameworks with continuous age-structured formalisms \cite{kynaston_2022}. More recently, Pang and Pardoux \cite{pang_2023} established formal connections between stochastic individual-based models, renewal equations, and age-structured PDEs. Despite these advances, most nonlinear age-structured PDEs remain analytically intractable. Explicit characterisations of steady states, division-time distributions, and invasion speeds remain rare, limiting the ability to derive interpretable predictions from such models.

We develop an analytical framework for collective cell invasion that links single-cell division timing to population persistence and invasion. Models are classified by their rate functions---constant, exponentially decaying, or gamma-like. Integration with respect to age yields a set of integro--differential equations where age structure enters through weighted integrals. This finite integro--differential formulation can be analysed directly, without closure approximations. Division and death are treated as separate processes, each dependent on cell age and local density. We treat non-spatial and spatial dynamics. In the non-spatial setting, we obtain explicit formulae for steady-state age distributions, cell-cycle-time distributions, and the parameter conditions for population persistence. In the spatial setting, we derive expressions for the minimal wave speed and recover the classical \FKPP{} limit when division and death are age-independent. To our knowledge, this is the first deterministic analysis within this class of nonlinear age-structured PDEs to provide explicit cell-cycle-time distributions together with survival conditions and invasion criteria. It addresses a gap previously approached only through stochastic or discrete-time models \cite{yates_2017, belluccini_2022}. We show that the shape of the division kernel controls both population persistence and invasion speed. Biologically, the results account for the gamma-shaped division-time distributions observed in lineage-tracing data and demonstrate how longer division times slow invasion. The framework links measurable quantities to the underlying division and death mechanisms. 

The paper is structured as follows. Section~2 introduces the age-structured model and derives the moment-hierarchy reduction. Section~3 analyses non-spatial dynamics and establishes survival thresholds that depend systematically on division kernel shape. Sections~4--5 characterise steady-state age distributions and cell-cycle-time distributions. Section~6 establishes minimal wave speeds for spatial models and shows that classical linear theory overestimates invasion rates when division is age-dependent. Section~7 discusses biological implications and experimental validation. Our approach builds on the work of Gurtin and MacCamy \cite{gurtin_1979, gurtin_1981} but diverges by separating division and death explicitly, retaining age-dependence in both loss terms, and deriving analytical results without closure assumptions.

\section{Model framework}\label{sec:framework}
We model a motile and proliferative cell population, where cell age tracks time since last division. Let \(u(a,x,t)\) represent the density of cells of age \(a \geq 0\) at position \(x\in\Omega\subseteq\mathbb{R}^d\) and time \(t\ge 0\), with total population
\begin{equation}\label{eq:totalpop}
P(x,t) = \int_0^\infty u(a,x,t)\,\D a.
\end{equation}
Cells age at unit rate, move diffusively with coefficient \(\kappa\ge 0\), and are removed by death at rate \(\mu(a,P)\ge 0\) or division at rate \(\beta(a,P)\ge 0\). Upon division, two daughter cells enter at age zero, giving the transport-renewal system
\begin{align}
\partial_t u + \partial_a u &= \kappa\,\Delta u - \bigl[\mu(a,P(x,t)) + \beta(a,P(x,t))\bigr]\,u, \label{eq:transport} \\
u(0,x,t) &= 2 \int_0^\infty \beta(a, P(x,t))\,u(a,x,t)\,\D a. \label{eq:renewal}
\end{align}
The boundary condition \eqref{eq:renewal} represents the influx of newborn cells. Initial data are \(u(a,x,0) = u_0(a,x) \ge 0\). For bounded \(\Omega\), we impose no-flux boundary conditions (confined tissue). For \(\Omega = \mathbb{R}^d\), one may assume decay as \(|x|\to\infty\) (unbounded spread).

The model captures two mechanisms. First, division and death rates depend on cell age and local density \(P(x,t)\). This density dependence captures contact inhibition without imposing an explicit carrying capacity. Second, division and death are treated as distinct processes. This separation enables direct analysis of how ageing and division timing independently influence population dynamics.

We analyse five model cases (Table~\ref{tab:model-cases}): constant rates (age-homogeneous), exponential decay (early division), and gamma-type profiles (maturation delay). These forms are chosen for analytical tractability and biological relevance. Each case is studied in both non-spatial (\(\kappa=0\)) and spatial (\(\kappa>0\)) settings.
\begin{table}[htb]
\centering
\caption{Model cases studied in both non-spatial and spatial settings}
\label{tab:model-cases}
\renewcommand{\arraystretch}{1.2}
\begin{tabular*}{0.65\textwidth}{@{\extracolsep{\fill}}cll@{}}
\toprule
Case & \(\mu(a,P)\) & \(\beta(a,P)\)\\
\midrule
1 & \(\mu\) & \(\beta(1-P)\) \\
2 & \(\mu\) & \(\beta e^{-\alpha a}(1-P)\)\\
3 & \(\mu\) & \(\beta a e^{-\alpha a}(1-P)\) \\
4 & \(\mu P\) & \(\beta a e^{-\alpha a}(1-P)\) \\
5 & \((\mu - \gamma a e^{-\alpha a})P\) & \(\beta a e^{-\alpha a}(1-P)\)\\
\bottomrule
\end{tabular*}
\end{table}

\subsection{Integro--differential reduction}\label{sec:reduction}
We now derive an integro--differential reduction of the age-structured system \eqref{eq:transport}--\eqref{eq:renewal}. 

Integrating Equation~\eqref{eq:transport} with respect to age, and assuming \(u(a,x,t)\to 0\) as \(a\to\infty\), we obtain the population balance
\begin{equation}\label{eq:P-evol}
\partial_t P = \kappa\,\Delta P(x,t) + B(x,t) - \int_0^\infty \bigl[\mu(a,P(x,t))+\beta(a,P(x,t))\bigr]u(a,x,t)\,\D a,
\end{equation}
where \(B(x,t)=u(0,x,t)\) is given by the renewal condition \eqref{eq:renewal}. When division and death rates depend only on the total population, the system closes in terms of \(P(x,t)\) alone; see Lemma~\ref{lem:class1}. This corresponds to logistic-type dynamics (cf. Case 1). 
\begin{lemma}\label{lem:class1}
    Suppose the division and death rates depend only on total density, i.e. \(\mu(a,P)=\mu(P)\) and \(\beta(a,P)=\beta(P)\), with \(\mu,\beta\ge 0\) and continuous. Under the renewal condition \eqref{eq:renewal}, the total population \(P(x,t)=\int_0^\infty u(a,x,t)\,\D a\) satisfies the reaction-diffusion equation
    \begin{equation}\label{eq:class1-closure}
        \partial_t P=\kappa\,\Delta P+\bigl[\beta(P)-\mu(P)\bigr]\,P.
    \end{equation}
    More generally, for affine feedback \(\beta(P)=\beta_0-\beta_1 P\) and \(\mu(P)=\mu_0+\mu_1 P\), 
    \begin{equation}\label{eq:affine-logistic}
    \partial_t P=\kappa\,\Delta P+\bigl[(\beta_0-\mu_0)-(\beta_1+\mu_1)P\bigr]\,P,
    \end{equation}
    where the linear growth rate is \(r=\beta_0-\mu_0\) and the carrying capacity is \(K=(\beta_0-\mu_0)/(\beta_1+\mu_1)\). Hence \(K>0\) requires \(\beta_0>\mu_0\) and \(\beta_1+\mu_1>0\).
\end{lemma}
\begin{proof}
    Integrate \eqref{eq:transport} with respect to \(a\) with \(u(\infty,x,t)=0\) to give
    \begin{equation}\label{eq:class1-proof}
    \partial_t P - u(0,x,t)=\kappa\,\Delta P-\int_0^\infty\bigl[\mu(P)+\beta(P)\bigr]\,u\,\D a.
    \end{equation}
    Since \(\mu,\beta\) depend only on \(P\), \(\int_0^\infty\bigl[\mu(P)+\beta(P)\bigr]u\,\D a=\bigl[\mu(P)+\beta(P)\bigr]P\). From \eqref{eq:renewal}, \(u(0,x,t)=2\beta(P)P\), and \(\int_0^\infty u\,\D a=P\). Substitution yields \eqref{eq:class1-closure}.
\end{proof}

For age-dependent rates, we introduce age-weighted moments that mirror the division kernel structures. For any fixed decay parameter \(\alpha > 0\), define
\begin{equation}\label{eq:CD-defs}
C(x,t;\alpha) = \int_0^\infty e^{-\alpha a} u(a,x,t)\,\D a, \qquad
D(x,t;\alpha) = \int_0^\infty a e^{-\alpha a} u(a,x,t)\,\D a.
\end{equation}

Multiplying Equation~\eqref{eq:transport} by \(e^{-\alpha a}\) and integrating by parts with respect to \(a\) gives
\begin{equation}\label{eq:C-evol}
\begin{aligned}
    \partial_t C &= \kappa \Delta C(x,t; \alpha) + B(x,t) - \alpha C(x,t;\alpha) \\
    &\quad\quad  - \int_0^\infty \left[\mu(a,P(x,t)) + \beta(a,P(x,t))\right] e^{-\alpha a} u(a,x,t)\,\D a.
\end{aligned}
\end{equation}

Similarly, multiplying Equation~\eqref{eq:transport}  by \(a e^{-\alpha a}\) and integrating gives
\begin{equation}\label{eq:D-evol}
\begin{aligned}
\partial_t D &= \kappa \Delta D(x,t; \alpha) + C(x,t;\alpha) - \alpha D(x,t;\alpha) \\
&\quad\quad - \int_0^\infty \left[\mu(a,P(x,t)) + \beta(a,P(x,t))\right] a e^{-\alpha a} u(a,x,t)\,\D a.
\end{aligned}
\end{equation}

These moments also appear directly in the boundary renewal term, which becomes case-dependent:
\begin{equation}\label{eq:B-cases}
B(x,t) =
\begin{cases}
2\beta(1 - P(x,t))\,P(x,t), & \text{Case 1}, \\
2\beta(1 - P(x,t))\,C(x,t;\alpha), & \text{Case 2}, \\
2\beta(1 - P(x,t))\,D(x,t;\alpha), & \text{Cases 3--5}.
\end{cases}
\end{equation}

To express the integral loss terms in Equation~\eqref{eq:C-evol} and Equation~\eqref{eq:D-evol}, i.e., \\ \(\int_0^\infty [\mu(a,P(x,t))+\beta(a,P(x,t))]\,e^{-\alpha a}u(a,x,t)\,\D a\) and
\(\int_0^\infty [\mu(a,P(x,t))+\beta(a,P(x,t))]\,a e^{-\alpha a}u(a,x,t)\,\D a\), we define higher-order moments. For any constant decay parameter \( \alpha > 0 \), define
\begin{equation}\label{eq:I-defs}
\begin{aligned}
I_0(t,x;\alpha) &= \int_0^\infty e^{-2\alpha a} u(a,x,t)\,\D a, &\quad
I_1(t,x;\alpha) &= \int_0^\infty a e^{-2\alpha a} u(a,x,t)\,\D a, \\
I_2(t,x;\alpha) &= \int_0^\infty a^2 e^{-2\alpha a} u(a,x,t)\,\D a.
\end{aligned}
\end{equation}
The evolution of \((P, C, D\)) involves these higher moments, so the system does not close. We use positivity constraints and integral inequalities to bound \(I_k\) and derive population thresholds. Equations \eqref{eq:P-evol}, \eqref{eq:C-evol}, and \eqref{eq:D-evol}, together with definitions \eqref{eq:I-defs} and case-specific \(B(x,t)\) in Equation~\eqref{eq:B-cases}, constitute the reduced integro--differential system. 

\subsection{Population persistence: survival conditions and bounds on population size} \label{sec:threshold}
We assume a well-mixed system and determine when populations sustain positive steady states and derive upper bounds on equilibrium size. 

The well-mixed system (\( \kappa = 0 \)) reduces to a transport-renewal model,
\begin{equation}\label{eq:transport-nodiff}
\begin{aligned}
\pdv{u}{t} + \pdv{u}{a} &= -\left[\mu(a,P(t)) + \beta(a,P(t))\right] u(a,t), \\
u(0,t) &= B(t) = 2 \int_0^\infty \beta(a,P(t))\,u(a,t)\,\D a, \\
P(t) &= \int_0^\infty u(a,t)\,\D a.
\end{aligned}
\end{equation}
If \(P(0) \leq 1\) then the total population \( P(t) \in [0,1], \forall t \) by construction. This system admits an explicit solution via characteristics,
\begin{equation}\label{eq:char-sol}
u(a,t) =
\begin{cases}
u_0(a - t)\,S(a - t,a,t;P), & a > t, \\
B(t - a)\,S(0,a,t - a;P), & 0 \le a \le t,
\end{cases}
\end{equation}
where the survival probability from age \( a_0 \) to \( a \) (i.e., the probability that a cell neither dies nor divides during this interval) is given by
\begin{equation}\label{eq:survival}
S(a_0,a,t;P) = \exp\left(-\int_{a_0}^a \left[\mu(a', P(a' - a_0 + t)) + \beta(a', P(a' - a_0 + t))\right]\,\D a'\right).
\end{equation}

In the long-time limit, we assume convergence to a steady state,
\begin{equation}
u(a,t) \to F(a), \qquad P(t) \to \bar{P},
\end{equation}
where \( \bar{P} \) is the total population at equilibrium. The steady-state distribution satisfies
\begin{equation}\label{eq:steady-ode}
F'(a) = -\left[\mu(a,\bar{P}) + \beta(a,\bar{P})\right] F(a),
\end{equation}
with \( F(0) = 2 \int_0^\infty \beta(a,\bar{P})\,F(a)\,\D a\) and \(\bar{P} = \int_0^\infty F(a)\,\D a.\) We define the steady-state survival function
\begin{equation}\label{eq:S-def}
S(a,\bar{P}) = \exp\left(-\int_0^a \left[\mu(a',\bar{P}) + \beta(a',\bar{P})\right]\,\D a'\right),
\end{equation}
so that the age distribution can be expressed explicitly as
\begin{equation}\label{eq:F-form}
F(a) = F(0)\,S(a,\bar{P}).
\end{equation}

Substituting Equation~\eqref{eq:F-form} into the boundary condition \eqref{eq:steady-ode} yields the integral constraint
\begin{equation}\label{eq:Pbar-condition}
1 = 2 \int_0^\infty \beta(a,\bar{P})\,S(a,\bar{P})\,\D a.
\end{equation}
This equation determines the possible steady states \(\bar{P}\). A positive solution \(\bar{P}>0\) indicates a viable, persistent population. Our moment reduction approach leads to explicit upper bounds \(P_c\), derived case by case below, such that any steady state satisfies \(\bar{P} \le P_c\). These bounds hold under the assumptions that \(\beta(a,P)\) is non-increasing in \(P\), \(\mu(a,P)\) is non-decreasing in \(P\), and the relevant moments \(C, D, I_k\) are finite.

\medskip
\paragraph{Case 1: Age-homogeneous dynamics}
When division and death rates are independent of age, all cells share the same survival probability regardless of time since division. We assume that 
\[
\mu(a,P)=\mu,\qquad \beta(a,P)=\beta(1-P),
\]
with constants \( \mu > 0 \), \( \beta > 0 \), where the term \(\beta(1-P)\) models contact inhibition: division potential decreases with density. Constant \(\mu\) implies death rate is the same at all densities.

The governing Equation~\eqref{eq:P-evol} reduces to logistic growth (Lemma~\ref{lem:class1}):
\[
P'(t) = (\beta - \mu)\,P(t)\left(1 - \frac{P(t)}{K}\right), \quad K = \frac{\beta - \mu}{\beta}.
\]
Survival requires \(\beta>\mu\); otherwise \(P\equiv0\) is the only steady state. When \(\beta>\mu\), the unique positive equilibrium is
\begin{equation}\label{eq:Pc1}
P_c = 1 - \frac{\mu}{\beta},
\end{equation}
representing the classical logistic balance between proliferation and mortality.

\medskip
\paragraph{Case 2: Early-division bias}
When cells have short cell cycle times (e.g., embryonic stem cells or intestinal epithelial cells), division events occur predominantly shortly after the previous division. We therefore assume that
\begin{equation}
\mu(a, P) = \mu, \qquad \beta(a, P) = \beta e^{-\alpha a}(1 - P),
\end{equation}
with constants \( \mu, \beta, \alpha > 0 \). The parameter \(\alpha\) controls the decline in division rate and favours division at younger ages. The contact inhibition factor \((1-P)\) limits division at high density. Inserting these rates into the main PDE, Equation \eqref{eq:transport}, and using the model reduction outlined in Section \ref{sec:reduction} yields a reduced moment system in the variables \((P,C,I_0)\), 
\begin{equation}
\begin{aligned}
    P'(t) &= -\mu P(t) + \beta(1 - P(t)) C(t), \\
    C'(t) &= -(\mu + \alpha) C(t) + \beta(1 - P(t))\left(2C(t) - I_0(t)\right).
\end{aligned}
\end{equation}
At steady state, algebraic relations between \(\bar{P}\), \(\bar{C}\), and \(I_0\) with the constraint \( \bar{I}_0 \ge 0 \) yields the survival condition \(2\beta > \mu + \alpha\). Any positive steady state satisfies
\begin{equation}\label{eq:Pc2}
P_c = 1 - \frac{\mu + \alpha}{2\beta}.
\end{equation}
Compared to Case 1, the survival condition reflects the combined effects of mortality (\(\mu\)) and age-related decline in division rates (\(\alpha\)).

\medskip
\paragraph{Case 3: Maturation delay}
When cells require time to pass cell-cycle checkpoints (G1/S/G2/M) before dividing, division potential is low immediately after division and peaks at intermediate ages. We therefore assume that
\[
\mu(a,P)=\mu,\qquad \beta(a,P)=\beta a e^{-\alpha a}(1-P),
\]
with constants \(\mu, \beta, \alpha > 0 \). The division rate follows a gamma-type profile, which peaks at age \(a=1/\alpha\). Inserting these rates into Equation \eqref{eq:transport} and using the reduction framework (Section~\ref{sec:reduction}) yields a coupled moment system in the variables \((P,C,D,I_1,I_2)\),
\begin{equation}
    \begin{aligned}
    P'(t) &= -\mu P(t) + \beta(1 - P(t))\,D(t), \\
    C'(t) &= -(\mu + \alpha)\,C(t) + \beta(1 - P(t))\left(2D(t) - I_1(t)\right), \\
    D'(t) &= -(\mu + \alpha)\,D(t) + C(t) - \beta(1 - P(t))\,I_2(t).
    \end{aligned}
\end{equation}
Solving this system at steady state, and imposing non-negativity of \((I_1, I_2)\) yields the survival condition \(2\beta>(\mu+\alpha)^2\). Any positive steady state satisfies
\begin{equation}\label{eq:Pc3}
P_c=1-\frac{(\mu+\alpha)^2}{2\beta}.
\end{equation}
The quadratic penalty reflects a compounded cost: cells must survive both background mortality (\(\mu\)) and the period of low division probability before reaching peak fertility at age \(1/\alpha\).

\medskip
\paragraph{Cases 4 and 5: Maturation delay with density-dependent death rate}
When cell death is driven by crowding rather than intrinsic processes, death rate scales with population density. We combine density-dependent death with the gamma-type division kernel from Case 3. In Case 4, we assume \(\mu(a,P)=\mu P\). Using the reduction framework (Section~\ref{sec:reduction}) yields coupled moment equations in \((P,C,D,I_1,I_2)\) analogous to those in Case 3. Again solving at steady state and assuming non-negativity of the higher moments yields a survival condition \(2\beta > \alpha^2\), and the population bound becomes
\begin{equation}\label{eq:Pc4}
P_c=\frac{-(\alpha\mu+\beta)+\sqrt{(\alpha\mu+\beta)^2-\mu^2(\alpha^2-2\beta)}}{\mu^2}.
\end{equation}
Case 5 extends this by incorporating age-dependent mortality. We assume \(\mu(a,P)=(\gamma - \mu a e^{-\alpha a})P\), which captures elevated death rates at birth, a decline during maturation, and a subsequent rise with senescence. Despite this added complexity, the survival condition remains \(2\beta>\alpha^2\), and the steady-state population is bounded above,
\begin{equation}
    \bar{P}\leq\min\left(P_c,\,\beta/(\gamma+\beta)\right),
\end{equation}
where \(P_c\) is given by Equation~\eqref{eq:Pc4}. The second constraint arises from moment sign requirements; for moderate \(\gamma\), the first term is the relevant upper bound (see Supplementary Materials, Section~3.3).

The critical shift from Cases 1--3 is that the \emph{survival condition} now depends only on the division kernel parameters \((\beta, \alpha)\), not on the death rate \(\mu\). At low density (\(P \to 0\)), density-dependent death terms vanish, so persistence is determined solely by the low-density division potential \(\beta(a,0)\). However, the equilibrium population level \(P_c\) does depend on \(\mu\). This reflects how the death rates modulate the carrying capacity but not the survival condition.

\medskip
These cases show a clear progression in how age structure constrains survival. Case 1 recovers the classical rule that net growth determines persistence. Case 2 imposes a linear penalty for early reproductive decline. Case 3 introduces a quadratic penalty that reflects the compounded challenge of surviving both a non-reproductive phase and background mortality. This progression quantifies the cost of delayed reproduction. In Cases 4--5, the survival condition depends only on division kernel shape, whereas death rate parameters affect the equilibrium bound \(P_c\), and determine how large the population can be once it survives.

Figure~\ref{fig:Fig1} and Table~\ref{tab:comparison-Pc} summarise these results. Numerical simulations confirm that steady-state populations respect their theoretical bounds \(P_c\). In Case 1, solutions match the bound exactly because the moment system closes in \(P\). Cases 2--3 track their respective bounds closely, while Cases 4--5 form bands beneath them. This banding stems from a fundamental difference in how mortality operates. In Cases 1--3, constant death rates affect growth at all densities. Both the survival condition and the equilibrium population density involve \(\mu\), which sets a unique equilibrium for a given \((\beta, \alpha, \mu)\). In Cases 4--5, density-dependent death vanishes at low densities. Survival is \(\mu\)-independent, depending only on \((\beta, \alpha)\), but the bound \(P_c\) is \(\mu\)-dependent. This gives one degree of freedom: varying \(\mu\) while holding \((\beta, \alpha)\) fixed produces one survival condition with different equilibria, hence the bands.

\begin{figure}[htb]
  \centering
  \includegraphics[width=\textwidth]{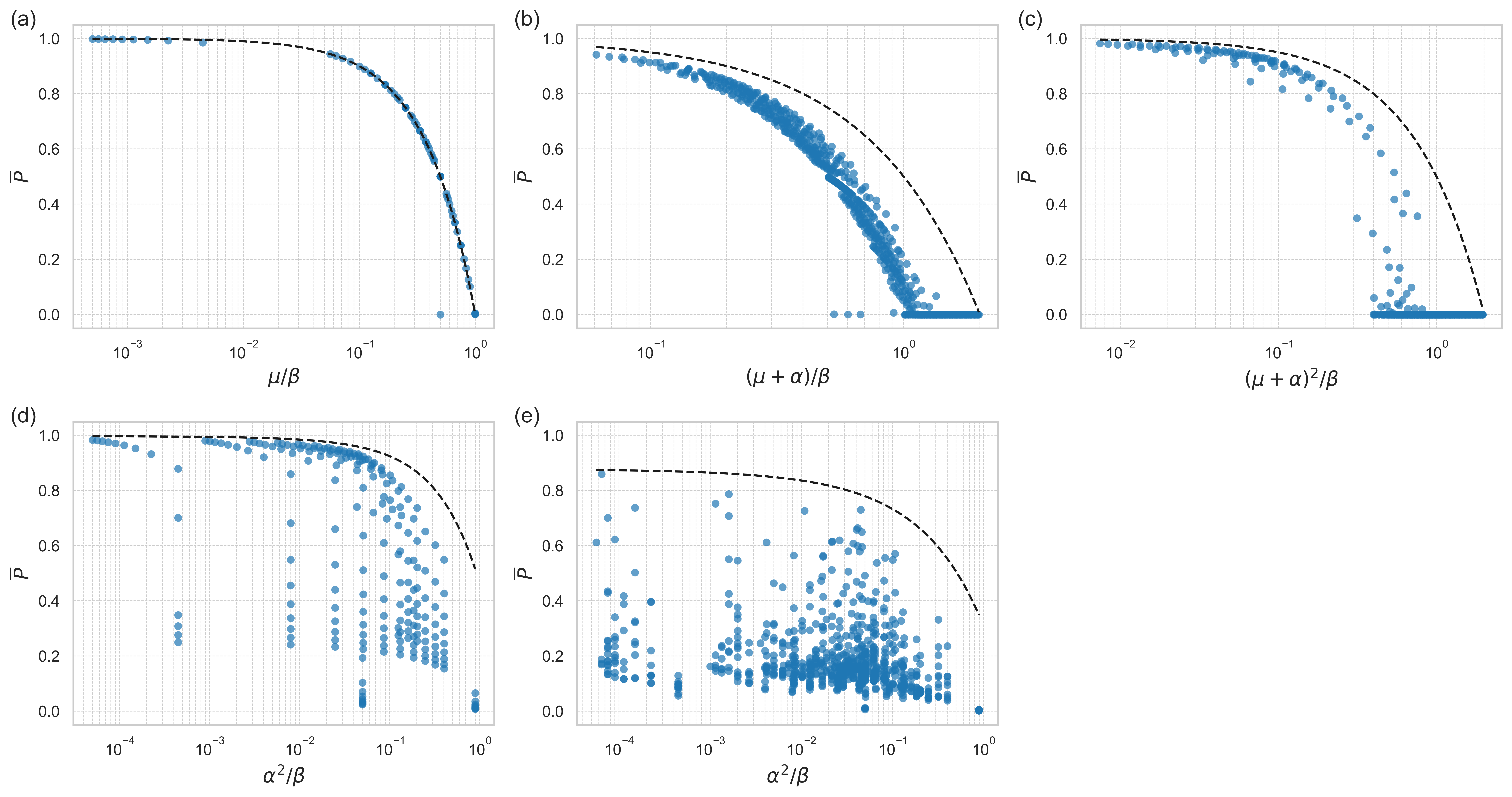}
  \caption{Numerical steady-state population sizes \( \bar{P} \) plotted against threshold conditions (Table~\ref{tab:comparison-Pc}) for each non-spatial model (Cases 1--5). Dashed lines indicate the analytical upper bounds \( \bar{P}_c \). Each panel corresponds to a different model case: (a) Case 1; (b) Case 2; (c) Case 3; (d) Case 4; and (e) Case 5.}
  \label{fig:Fig1}
\end{figure}

\begin{table}[htb]
\centering
\caption{Summary of non-spatial models: population bounds \(P_c\) (with \(\bar{P}\le P_c\)) and survival conditions admitting positive steady states.}
\label{tab:comparison-Pc}
\renewcommand{\arraystretch}{1.3}
\begin{tabular*}{\textwidth}{@{\extracolsep{\fill}}cll@{}}
\toprule
Case & Upper bound \( P_c \) & Survival condition \\
\midrule
1  & \( 1 - \mu/\beta \) & \( \beta > \mu \) \\
2 & \( 1 - \frac{\mu + \alpha}{2\beta} \) & \( 2\beta > \mu + \alpha \) \\
3 & \( 1 - \frac{(\mu + \alpha)^2}{2\beta} \) & \( 2\beta > (\mu + \alpha)^2 \) \\
4 &
\(
\frac{- (\alpha \mu + \beta) + \sqrt{(\alpha \mu + \beta)^2 - \mu^2 (\alpha^2 - 2\beta)} }{\mu^2}
\)
& \( 2\beta > \alpha^2 \) \\
5&
\(
\min\left(\frac{-(\alpha\mu+\beta) + \sqrt{(\alpha\mu+\beta)^2 - \mu^2(\alpha^2-2\beta)}}{\mu^2}, \frac{\beta}{\gamma+\beta}\right)
\)& \( 2\beta > \alpha^2,\, \mu > \gamma/\alpha e \) \\
\bottomrule
\end{tabular*}
\end{table}

\section{Steady-state age distributions}\label{sec:ssd}
As established in Section~\ref{sec:threshold}, the steady-state age distribution (SSAD) takes the form \(F(a)=F(0)\,S(a;\bar P)\), where \(S(a;\bar P)\) is the survival function and \(F(0)\) is determined by the renewal constraint \eqref{eq:Pbar-condition}. It describes the fraction of cells of age \(a\) in a population at equilibrium. Intuitively, it can be expressed as
\begin{equation}
    \begin{aligned}
        F(a) = & \text{ renewal factor } \times \mathbb{P}(\text{cell survives beyond age } a) \\
         & \times \mathbb{P}(\text{cell does not divide by age } a),
    \end{aligned}
\end{equation}
which captures the balance between survival and renewal across ages. A useful summary measure is the mean population age,
\begin{equation}
\bar{a}_\mathrm{pop} = \frac{\int_0^{\infty} a F(a)\D a}{\int_0^{\infty} F(a)\D a},
\end{equation}
which quantifies the average time since last division for cells in the population at steady state.

Table~\ref{tab:comparison-Fa} lists explicit formulae for \(F(a)\) across all five cases. The value of \(\bar P\) is found by solving the integral constraint \eqref{eq:Pbar-condition} numerically. Once \(\bar P\) is known, either from this computation or from data, these formulae give \(F(a)\) directly without integrating the full PDE specified in Equation~\eqref{eq:transport-nodiff}.

Figure~\ref{fig:Fig2} validates the analytical formulae. Panels (a--e) show numerical solutions (solid) match analytical predictions (dashed) across all cases. Panel (h) reveals a consistent pattern: mean division age \(\bar{a}_{\mathrm{div}}\) falls below mean population age \(\bar{a}_{\mathrm{pop}}\) in all cases. This occurs because divisions concentrate among cells in the reproductive window, while the population includes older, non-dividing survivors that persist in the tail and increase the mean age. The explicit SSADs therefore quantify how division timing shapes population structure. Age-independent rates (Cases 1--2) produce exponential-decay profiles because the division rates are constant across all ages. In contrast, gamma-type kernels (Cases 3–5) produce concave profiles near the origin. Newborn cells have a near-zero division probability, which causes them to accumulate into intermediate age classes before dividing.

\begin{table}[htb]
    \centering
    \caption{Summary of steady-state age distributions \( F(a) \). Cases 1-5 correspond to different division and death rate structures as defined in Table~\ref{tab:model-cases}.}
    \label{tab:comparison-Fa}
    \renewcommand{\arraystretch}{1.3}
    \begin{tabular*}{0.85\textwidth}{@{\extracolsep{\fill}}ccl@{}}
        \toprule
        Case & \( F(0) \) & \( F(a) = F(0)\,e^{-\int_0^a h(a')\,\D a'} \)\\ 
        \midrule
        1 & \( 2\beta(1 - \bar{P})\bar{P} \) & \( F(0)\, e^{-\mu a - \beta a (1 - \bar{P})} \) \\
        2 & \( 2\beta(1 - \bar{P})\bar{C} \) &  \( F(0)\,e^{-\mu a - \tfrac{\beta(1 - \bar{P})}{\alpha}(1 - e^{-\alpha a})} \) \\
        3 & \( 2\beta(1 - \bar{P})\bar{D} \) & \( F(0)\,e^{-\mu a - \tfrac{\beta(1 - \bar{P})}{\alpha^2}[1 - e^{-\alpha a}(1 + \alpha a)]} \) \\
        4 & \( 2\beta(1 - \bar{P})\bar{D} \) & \( F(0)\,e^{-\mu a \bar{P} - \tfrac{\beta(1 - \bar{P})}{\alpha^2}[1 - e^{-\alpha a}(1 + \alpha a)]} \) \\
        5 & \( 2\beta(1 - \bar{P})\bar{D} \) & \( F(0)\,e^{-\mu a\,\bar{P} + \tfrac{\gamma\,\bar{P}}{\alpha^2} [1 - e^{-\alpha a}(1 + \alpha a)] - \tfrac{\beta(1 - \bar{P})}{\alpha^2} [1 - e^{-\alpha a}(1 + \alpha a)]} \) \\
        \bottomrule
    \end{tabular*}
\end{table}

\begin{figure}[htb]
  \centering
  \includegraphics[width=\textwidth]{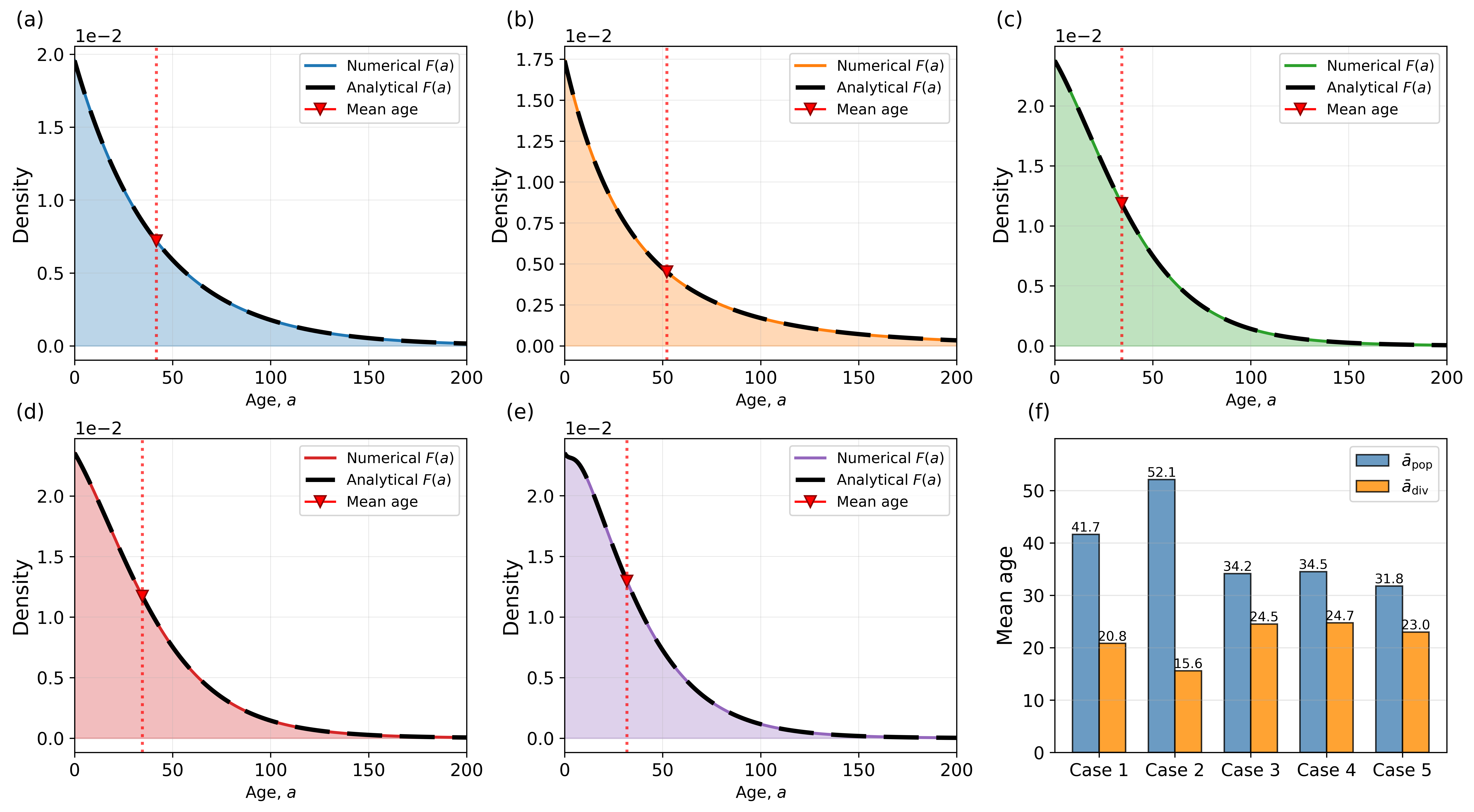}
  \caption{Steady-state age distributions across five model cases with different division and death structures. Panels (a)-(e) show steady-state distributions $F(a)$ for Cases 1-5: numerical solutions (coloured lines with shading) compared to analytical  (black dashed lines). Red triangles mark the mean population age for each case. Panel (f) compares mean population age $\bar{a}_{\mathrm{pop}}$ and mean division age $\bar{a}_{\mathrm{div}}$ across all cases.}
  \label{fig:Fig2}
\end{figure}

\section{Cell-cycle-time distributions}\label{sec:cctd}
The cell-cycle-time distribution (CCTD) describes the timing of a cell's first division in a population at steady-state. It complements the steady-state age distribution (SSAD) \(F(a)\), which describes population structure. Both derive from the survival probability \(S(a,\bar P)\) \eqref{eq:S-def}, but represent different biological quantities. The CCTD can be interpreted heuristically as
\begin{equation}
    \begin{aligned}
        f(a) = & \text{ division rate} \times \mathbb{P}(\text{cell survives beyond age } a) \\
         & \times \mathbb{P}(\text{cell does not divide by age } a),
    \end{aligned}
\end{equation}
which expresses the probability that a cell divides at age \(a\). We formalise this by defining the unnormalised cell division density as as \(\tilde{f}(a) = \beta(a, \bar{P}) S(a, \bar{P})\). This represents the density of cells undergoing their first division at age \(a\), given by the fraction of survivors \(S(a,\bar{P})\) (cells that have survived without dying or dividing until age \(a\)) multiplied by their instantaneous division rate \(\beta(a,\bar{P})\). The normalised CCTD is then
\begin{equation}
f(a) = \frac{\beta(a, \bar{P}) S(a, \bar{P})}{\int_0^\infty \beta(a', \bar{P}) S(a', \bar{P})\,\D a'}, \label{eq:CCTD}
\end{equation}
which ensures \(\int_0^\infty f(a) \, \D a= 1\). Using the SSAD \(F(a)=F(0) S(a,\bar P)\), the CCTD relates directly to the population structure:
\begin{equation}\label{eq:F-CCTD-relation}
f(a)=\frac{\beta(a,\bar P) F(a)}{\int_0^\infty \beta(a',\bar P) F(a')\,\D a'} = \frac{2 \beta(a,\bar P) F(a)}{F(0)},
\end{equation}
where the renewal condition \(F(0)=2\int_0^\infty \beta(a,\bar P) F(a)\,\D a\) gives the final form. Thus, \(f(a)\) is the division-weighted SSAD. The SSAD \(F(a)\) includes all cells, while the CCTD \(f(a)\) conditions on those that successfully divide.

A useful summary statistic is the mean division age defined by
\begin{equation}
    \bar{a}_\mathrm{div} = \int_0^\infty a f(a) \D a,
\end{equation}
which quantifies the expected timing of division events in the population.

The explicit forms of the CCTD for each case are given in Table~\ref{tab:comparison-CCTD}; Cases 1--2 produce exponential-like profiles, while Cases 3--5 yield gamma-type distributions with delayed peaks. These formulae provide a direct link to experiments. One can fit them to lineage-tracing data to either infer the underlying division kernel shape or, if the kernel form is known, to estimate specific parameters like the peak division age \((\alpha^{-1})\) and the maximum division rate \((\beta)\). This enables quantitative prediction of how division timing shapes population dynamics without conducting extensive numerical simulations. We first validate these expressions against numerical solutions (Fig.~\ref{fig:Fig3}) and then demonstrate their application to experimental data (Fig.~\ref{fig:Fig4}).

\begin{table}[htb]
    \centering
    \caption{Cell-cycle-time distributions (CCTDs) for each model case. The distribution is given by \( f(a) = \beta(a, \bar{P}) S(a, \bar{P}) / Z \), where \( Z \) is the normalization constant. The survival function is \( S(a, \bar{P}) = \exp\left(-\int_0^a h(a') da'\right) \) with the removal rate \( h(a) = \mu(a, \bar{P}) + \beta(a, \bar{P}) \).}
    \label{tab:comparison-CCTD}
    \renewcommand{\arraystretch}{1.2}
    \begin{tabular*}{0.75\textwidth}{@{\extracolsep{\fill}}cl@{}}
        \toprule
        Case & CCTD \( f(a) \) \\
        \midrule
        1 & 
        \( \beta(a,\bar{P})\, e^{ -(\mu + \beta(1 - \bar{P}))\,a } \) \\
        2 & 
        \( \beta(a,\bar{P})\, e^{ -\mu a - \frac{ \beta(1 - \bar{P}) }{ \alpha } [ 1 - e^{-\alpha a} ] } \) \\
        3 & 
        \( \beta(a,\bar{P})\, e^{ -\mu a - \frac{ \beta(1 - \bar{P}) }{ \alpha^2 } [ 1 - e^{-\alpha a}(1 + \alpha a) ] } \) \\
        4 & 
        \( \beta(a,\bar{P})\, e^{ -\mu a \bar{P} - \frac{ \beta(1 - \bar{P}) }{ \alpha^2 } [ 1 - e^{-\alpha a}(1 + \alpha a) ] } \) \\
        5 & 
        \( \beta(a,\bar{P})\, e^{ -\mu a \bar{P} + \frac{ \gamma \bar{P} }{ \alpha^2 } [ 1 - e^{-\alpha a}(1 + \alpha a) ] - \frac{ \beta(1 - \bar{P}) }{ \alpha^2 } [ 1 - e^{-\alpha a}(1 + \alpha a) ] } \) \\
        \bottomrule 
    \end{tabular*}
\end{table}
Figure~\ref{fig:Fig3} illustrates how the CCTD expressions translate into observable cell division-timing patterns. Cases 1--2 concentrate divisions at young ages. The exponential form in Case 1 (\(k=1.0\)) reflects memoryless division timing, where cells have a constant division rate regardless of age, consistent with classical birth-death models. In Case 2 (\(k=0.95\)), division potential decreases exponentially with age, producing a hypoexponential distribution that shifts divisions to even earlier ages. This captures fast-cycling regimes where many cells divide soon after birth. In contrast, Cases 3--5 incorporate gamma-type division kernels \(\beta(a,\bar P) \propto a e^{-\alpha a}\), which model a maturation delay and produce the bell-shaped CCTDs observed in lineage-tracing data \cite{smith_1973, leander_2014}. This captures the more realistic scenario where newly divided cells require time to prepare for the next division cycle. The peak location is set by \(a \approx 1/\alpha\), while the distribution height and width are determined by \(\beta\) and the death rate \(\mu\). All cases exhibit an exponential tail \(\sim e^{-\mu(a,\bar{P})}\). Panel (f) overlays all normalised cases, showing the ordering implied by their means: exponential-decay division rates (Case 2) yield the lowest mean division age (15.6), followed by constant rates (Case 1, 20.8), then gamma-type cases (3--5) with later peaks and higher mean ages (23--25).

We fitted the analytical CCTD for Case 3 to lineage data from cells under control (DMSO) and drug-perturbed conditions (Fig.~\ref{fig:Fig4}(a)--(c)). CHX (cycloheximide) is a protein synthesis inhibitor that lengthens the G1 phase, and erlotinib (ERL) is an EGFR inhibitor that extends all cycle phases. Case 3 was selected because its gamma-type division rate is the simplest form that produces the unimodal CCTD observed experimentally. The minimum cell-cycle time \(a_0\) was set to the 10th percentile of the data, fixing the left-hand support of the distribution. For each \(\alpha\) in a plausible range, we estimated \((\beta,\mu)\) by minimising a combined loss: the \(L^2\) difference between model and empirical cumulative distribution functions (CDFs), plus a weighting on tail survival probabilities to ensure accurate long-time behaviour. The survival condition \(2\beta > (\mu+\alpha)^2\) was enforced via quadratic penalty. The fitted parameters \((\beta, \alpha, \mu)\) reflect the expected mechanisms: CHX and erlotinib reduce the division rate (\(\beta\)) and increase the maturation timescale (\(1/\alpha\)) compared to the DMSO control. Full fitting details and code are provided in the Supplementary Materials.
\begin{figure}[htb]
  \centering
  \includegraphics[width=\textwidth]{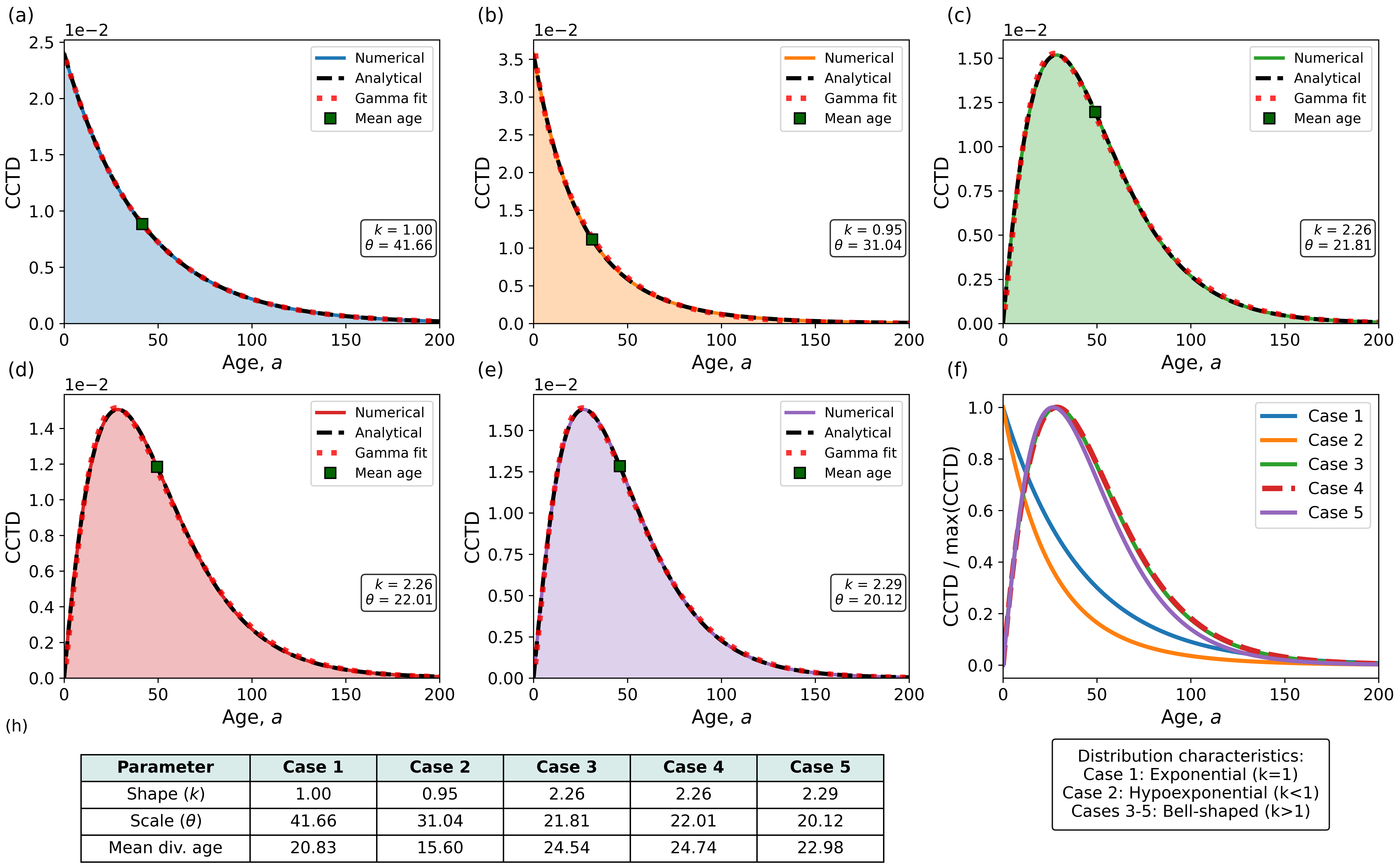}
  \caption{Cell–cycle-time distributions across five model cases. Panels (a--e) show numerical CCTDs (solid) with fitted gamma distributions (dashed) for Cases 1--5; the fitted shape \(k\) and scale \(\theta\) parameters are indicated. Squares mark mean division ages. Panel (f) overlays all cases after normalisation to isolate shape differences: exponential (Case 1, \(k=1\)), hypoexponential (Case 2, \(k<1\)), and bell–shaped (Cases 3--5, \(k>1\)). Panel (h) summarises the fitted gamma parameters and mean division ages.}
  \label{fig:Fig3}
\end{figure}

\begin{figure}[htb]
  \centering
  \includegraphics[width=\textwidth]{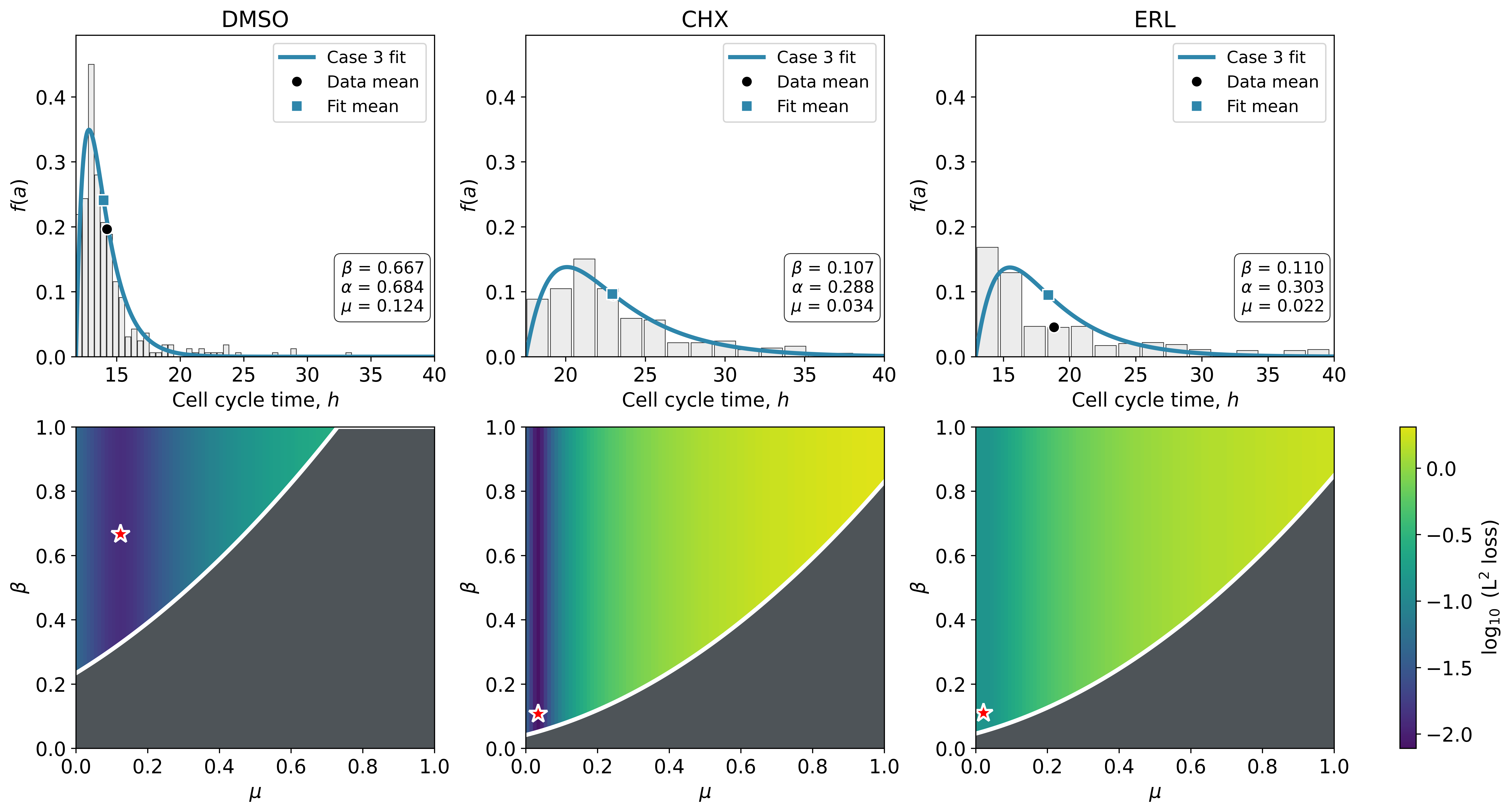}
  \caption{Experimental cell–cycle time distributions with Case~3 distributions and gamma fits. Panels (a--c) show CCTDs under DMSO (vehicle control), cycloheximide (CHX), and erlotinib (ERL). Data are plotted with the Case~3 analytical CCTD fit (solid) and a gamma distribution fit (dashed). Insets report the fitted \((\beta,\alpha,\mu)\); squares mark mean division age. Panels (d--f) show the corresponding loss landscapes in \((\mu,\beta)\) at the fitted \(\alpha\) on a shared colour scale: the white curve is the constraint boundary \(\beta=\frac{1}{2}(\mu+\alpha)^2\) (i.e., parameter condition for survival per Table~\ref{tab:comparison-Pc}); the red star marks the optimum; infeasible regions are shaded. Axes are \(\mu\in[0,1]\) and \(\beta\in[0,1]\). The infeasible region differs across panels because the feasibility boundary depends on the fitted \(\alpha\) for each dataset.}
  \label{fig:Fig4}
\end{figure}

\section{Travelling waves and invasion thresholds}\label{sec:travelling-waves}
We now examine travelling-wave solutions of the spatial (\(\kappa>0\)) models to assess how age structure modifies invasion speed. Classical \FKPP{} theory predicts a single wave speed determined by the net growth rate. In the context of age-structured models, two distinct linear predictions emerge. The first, \(c_{\mathrm{lin}}\), arises from the moment-system approximation and coincides with the \FKPP{} prediction. The second, \(c_{\min}\), obtained from the age-structured formulation, is systematically smaller. This discrepancy arises because the moment-system linearisation assumes uniform decay across all ages at the wave front, whereas the age-structured model preserves differential decay where younger cells dominate the front, and their lower division potential slows invasion. Moreover, the condition for a positive wave speed exactly matches the non-spatial survival condition, linking persistence and invasion within a unified framework.
\subsection{Moment-system linearisation and the speed \texorpdfstring{\(c_{\mathrm{lin}}\)}{c\_lin}}
We seek solutions of the form
\begingroup
\setlength{\abovedisplayskip}{4pt}
\setlength{\belowdisplayskip}{4pt}
\begin{equation}\label{eq:TW-ansatz}
u(a,x,t) = \phi(a,z),
\end{equation}
\endgroup
with \(z = x - ct\) and \(c > 0\), representing a front propagating at constant speed \(c\). We apply the travelling wave ansatz~\eqref{eq:TW-ansatz} to the reduced integro--differential system \eqref{eq:P-evol}--\eqref{eq:I-defs}, expressing the population moments as functions of \(z\),
\begin{align}
U(z) &= \int_0^\infty \phi(a,z) \,\D a, \quad 
V(z;\alpha) = \int_0^\infty e^{-\alpha a}\phi(a,z) \,\D a, \label{eq:UVW-def}\\
W(z;\alpha) &= \int_0^\infty a e^{-\alpha a}\phi(a,z) \,\D a, \quad
\tilde{I}_k(z;\alpha) = \int_0^\infty a^k e^{-2\alpha a} \phi(a,z) \,\D a \quad \text{for } k = 0, 1, 2. 
\end{align}
The parameter \(\alpha\) here represents the rate at which fertility declines with age in the division kernel. We linearise the resulting system about the extinction state where all moments vanish, i.e. \(\;U(z)=V(z;\alpha)=W(z;\alpha)=\tilde I_k(z;\alpha)=0\;\) for \(k=0,1,2\). This is mathematically tractable and makes the specific biological assumption that all age classes vanish uniformly at the wave front. For each model case, this procedure yields an explicit expression for the wave speed, which we denote \(c_{\mathrm{lin}}\). For age-independent rates (Case 1), this gives the classical \FKPP{} speed, \(c_{\mathrm{lin}} = 2\sqrt{\kappa r}\) with \(r=\beta - \mu\). For models with age structure, \(c_{\mathrm{lin}}\) takes a more complex form but maintains the general structure \(2\sqrt{\kappa~r_{\mathrm{eff}}}\), where \(r_{\mathrm{eff}}\) is an effective low-density growth rate.

Table~\ref{tab:wave-speed-summary} summarises the wave speeds \(c_{\mathrm{lin}}\) for Cases 1--5. Numerical simulations (Fig.~\ref{fig:Fig5}) confirm that \(c_{\mathrm{lin}}\) matches the observed invasion speed when division and death rates are age-independent (Case 1). For age-dependent rates, however, \(c_{\mathrm{lin}}\) overestimates the observed speed because the age distribution at the wave front is not uniform. We derive the exact minimal speed \(c_{\min}\) in the next section and show \(c_{\min} \leq c_{\mathrm{lin}}\) in all cases.

\subsection{Age-structured analysis and minimal wave speeds}
We apply the ansatz~\eqref{eq:TW-ansatz} to the governing PDE \eqref{eq:transport},
\begin{equation}\label{eq:linearised-pde0}
-c\phi_z + \phi_a = \kappa \phi_{zz} - [\mu(a,U(z)) + \beta(a,U(z))]\phi(a,z).
\end{equation}
At the leading edge, where \(U(z) \to 0\), the linearised system becomes
\begin{equation}\label{eq:linearised-pde}
-c\phi_z + \phi_a = \kappa \phi_{zz} - [\mu(a,0) + \beta(a,0)]\phi(a,z),
\end{equation}
with boundary condition
\begin{equation}\label{eq:linearised-bc}
\phi(0,z) = 2\int_0^\infty \beta(a,0)\phi(a,z) \,\D a.
\end{equation}
We seek solutions of the form \(\phi(a,z) = A(a)e^{-\lambda z}\), where \(\lambda > 0\) determines the spatial decay rate. Substituting into \eqref{eq:linearised-pde} yields the age profile equation
\begin{equation}
A'(a) + \left[c\lambda - \kappa\lambda^2 + \mu(a,0) + \beta(a,0)\right]A(a) = 0,
\end{equation}
with solution
\begin{equation}\label{eq:age-profile}
A(a) = A(0)\exp\left[-\int_0^a \left(c\lambda - \kappa\lambda^2 + \mu(s,0) + \beta(s,0)\right) \,\D s\right].
\end{equation}
Substituting \eqref{eq:age-profile} into \eqref{eq:linearised-bc} gives the dispersion relation
\begin{equation}\label{eq:general-dispersion}
1 = 2\int_0^\infty \beta(a,0) \exp\left[-\int_0^a (c\lambda - \kappa\lambda^2 + \mu(s,0) + \beta(s,0)) \,\D s\right] \,\D a.
\end{equation}
The minimal wave speed \(c_{\min}\) is the smallest \(c > 0\) for which there exists \(\lambda > 0\) satisfying \eqref{eq:general-dispersion}. For tractable rate forms, \eqref{eq:general-dispersion} yields \(c_{\min}\) analytically. Otherwise, it reduces to a root-finding problem that is straightforward to compute numerically.

\subsection{Minimal wave speeds across model cases} \label{sec:TW-examples} 
We now apply both linearisation approaches to Cases 1 and 5, deriving \(c_{\min}\) explicitly and showing that \(c_{\min} = c_{\mathrm{lin}}\) for Case 1 but \(c_{\min} < c_{\mathrm{lin}}\) for Case 5.

\paragraph{Case 1: Age-homogeneous dynamics}
Consider the system \eqref{eq:transport}--\eqref{eq:renewal} with
\begin{equation}
\mu(a,P) = \mu, \qquad \beta(a,P) = \beta(1 - P).
\end{equation}
This case serves as a benchmark where age structure decouples from invasion dynamics. We recall that the non-spatial survival condition is \(\beta>\mu\).

Integrating the PDE \eqref{eq:transport} with respect to \(a\) and using the boundary condition \( u(0,x,t) = 2\beta(1 - P)P \) gives (Lemma~\ref{lem:class1})
\begin{equation}
\partial_t P = \kappa \Delta P + r P\left(1 - \frac{P}{K}\right). \label{eq:FKPP-form}
\end{equation}
where \(r=\beta-\mu\) and \(K=(\beta-\mu)/\beta\). Applying the travelling-wave ansatz \(P(x,t) = U(z)\), \(z=x-ct\), and linearising about \(U=0\) yields
\begin{equation}
\kappa\lambda^2 - c\lambda + (\beta - \mu) = 0. \label{eq:FKPP-dispersion}
\end{equation}
This dispersion relation is identical to that obtained from the general formula \eqref{eq:general-dispersion} with \(\mu(a,0) = \mu\) and \(\beta(a,0) = \beta\), confirming that the moment-system and age-structured linearisations agree. Hence
\begin{equation}
c_{\min} = c_\mathrm{lin} = 2\sqrt{\kappa(\beta-\mu)}. \label{eq:cmin-1}
\end{equation}
If \(\beta > \mu\) the population invades; otherwise the population goes extinct. This matches the non-spatial survival condition. Invasion speed depends solely on the net growth rate \(\beta - \mu\), not age structure, linking persistence to spatial spread.

\medskip
\paragraph{Case 1b: Density-dependent death rate}
Consider a modification: \(\mu(a,P)=\mu P\) and \(\beta(a,P)=\beta(1-P)\). Lemma~\ref{lem:class1} applies, giving
\begin{equation}\label{eq:Case1b-PDE}
\partial_t P=\kappa \Delta P+\beta\,P-(\beta+\mu)P^2.
\end{equation}
This is again a \FKPP{}-type equation with \(r=\beta\) and \(K=\beta/(\beta+\mu)\). Since the division and death rates remain age-independent, we have \(c_{\min} = c_{\mathrm{lin}}\), with
\begin{equation}\label{eq:cmin-Case1b}
c_{\min} = 2\sqrt{\kappa\,\beta}.
\end{equation}
The key difference from Case 1 lies in how the death rate affects invasion. In Case 1, the death rate is independent of the density (\(\mu(a,P) = \mu\)), reducing the net growth rate to \(r = \beta - \mu\). In Case 1b, the death rate is density-dependent (\(\mu(a,P) = \mu P\)) and vanishes at the wave front where \(P \to 0\). The invasion speed is therefore governed by the division rate at low density, while the death rate regulates the carrying capacity behind the front.

\medskip
\paragraph{Case 5: Maturation delay with density-dependent death rate} Consider
\begin{equation}
\mu(a,P) = (\mu - \gamma a e^{-\alpha a})\,P, \qquad 
\beta(a,P) = \beta a e^{-\alpha a}(1 - P),
\end{equation}
where \( \mu > \gamma/(\alpha e) \) ensures the death rate remains positive, and the non-spatial threshold \(2\beta > \alpha^2\) guarantees population persistence.

Using the reduction framework of Section~\ref{sec:reduction} and the ansatz~\eqref{eq:TW-ansatz}, the moment system in travelling-wave coordinates becomes
\begin{equation}
\begin{aligned}
-c U'&=\kappa U''-\mu U^{2}+ \gamma  U W + \beta (1-U) W,\\
-c V'&=\kappa V''-(\mu U+\alpha) V+ \gamma U \tilde{I_1} + \beta (1-U)\,\bigl(2W-\tilde{I}_1\bigr),\\
-c W'&=\kappa W''-(\mu U+\alpha) W+V+ \gamma U \tilde{I_2} -\beta (1-U) \tilde{I}_2,
\end{aligned}
\end{equation}
where the \(\tilde{I}_k\) are defined as in \eqref{eq:UVW-def}. Linearising about \((U, V, W) = (0,0,0)\), where higher moments vanish because \(\phi(a,z) \to 0\), yields
\begin{equation}
c_{\mathrm{lin}} = 2\sqrt{\kappa\left(\sqrt{2\beta} - \alpha\right)}.
\end{equation}
If\(2\beta > \alpha^2\), which matches the condition for population persistence, the population invades; otherwise the population dies out. Unlike Case 1, however, \(c_{\mathrm{lin}}\) provides only an upper bound. To obtain the exact minimal speed \(c_{\min}\), we solve the dispersion relation \eqref{eq:general-dispersion}. At the wave front where \(P \to 0\), \(\mu(a,0) = 0\) and \(\beta(a,0) = \beta a e^{-\alpha a}\), giving
\begin{equation}
1 = 2\beta \int_0^\infty a e^{-\alpha a} \exp\left[-(c\lambda - \kappa\lambda^2)a - \frac{\beta}{\alpha^2}\left(1 - e^{-\alpha a}(1 + \alpha a)\right)\right] \D a. \label{eq:dispersion-case5}
\end{equation}
The minimal speed \(c_{\min}\) is obtained by minimising \(c(\lambda)\) over \(\lambda > 0\) satisfying \eqref{eq:dispersion-case5}. Numerical solution (Fig.~\ref{fig:Fig5}) confirms \(c_{\min} < c_\mathrm{lin}\).

\medskip
In summary, the upper bound \(c_{\mathrm{lin}}\) arises from linearising the moment system at the extinction state, assuming uniform decay across all age classes. The exact minimal speed \(c_{\min}\) accounts for the age-structured profile at the leading edge and satisfies \(c_{\min} \leq c_{\mathrm{lin}}\), with equality only when division and death rates are age-independent. The invasion conditions are parameter constraints derived by requiring \(c_{\mathrm{lin}} > 0\) and coincide with the non-spatial survival thresholds for all cases (Table~\ref{tab:comparison-Pc} and Table~\ref{tab:wave-speed-summary}). When invasion conditions hold, travelling waves exist with positive minimal speed; when they fail, populations cannot invade regardless of initial spatial configuration, as the population dies out locally before spatial spread can occur.

\begin{table}[htb]
    \centering
    \caption{Summary of invasion speed bounds and conditions for the spatial models. The upper bound \(c_\mathrm{lin} \) comes from moment-system linearisation, and the minimal speed \(c_{\min}\) satisfies \(c_{\min} \leq c_\mathrm{lin} \). Invasion conditions give the parameter regime where travelling waves exist (equivalently, where \(c_{\mathrm{lin}} > 0\)).}
    \label{tab:wave-speed-summary}
    \renewcommand{\arraystretch}{1.3}
    \begin{tabular*}{\textwidth}{@{\extracolsep{\fill}}clll@{}}
        \toprule
        Spatial case & Invasion speed bound \( c_\mathrm{lin} \) & Invasion condition \\
        \midrule
        1  &  \(2\sqrt{\kappa(\beta-\mu)}\) & \(\beta > \mu\) \\
        1b &  \(2\sqrt{\kappa\beta}\) & \(\beta > 0\) \\
        2  &  \(2\sqrt{\kappa(2\beta-\mu-\alpha)}\) & \(2\beta > \mu+\alpha\) \\
        3  & \(2\sqrt{\kappa(\sqrt{2\beta}-\mu-\alpha)}\) & \(2\beta > (\mu+\alpha)^2\) \\
        4  &  \(2\sqrt{\kappa(\sqrt{2\beta} - \alpha)}\) & \(2\beta > \alpha^2\) \\
        5  &  \(2\sqrt{\kappa(\sqrt{2\beta} - \alpha)}\) & \(2\beta > \alpha^2\) \\
        \bottomrule
    \end{tabular*}
\end{table}
Numerical solutions of the governing PDE~\eqref{eq:transport} converge to travelling waves with speeds matching theoretical predictions (Fig.~\ref{fig:Fig5}). For age-independent rates (Case 1), \(c_{\mathrm{min}} = c_{\mathrm{lin}}\). For age-dependent rates (Cases 3, 5), \(c_{\mathrm{min}} < c_{\mathrm{lin}}\). Figure~\ref{fig:Fig6} reveals the mechanism. The top row shows where cell division occurs spatially. In Case 1, division concentrates at the leading edge. In Cases 3 and 5, cell division concentrates slightly behind the leading edge. 
Consequently, peak division occurs not at the leading edge but at positions where cells have aged sufficiently to reach peak division rates. The bottom row shows death loss and division loss. In Case 1, high division rates at the front drive rapid invasion. In Cases 3 and 5, lower division rates at the front slow the wave.

The inequality \(c_{\min} \leq c_{\mathrm{lin}}\) quantifies how age structure affects invasion speed. The bound \(c_{\mathrm{lin}}\) assumes all ages contribute equally at the wave front and yields a tractable expression that depends only on division and death rates at low density. The exact minimal speed \(c_{\min}\) accounts for the age distribution at the leading edge but requires solving a dispersion relation numerically. Computing both quantities provides a diagnostic: if \(c_{\min} \approx c_{\mathrm{lin}}\), age structure is negligible; if \(c_{\min} \ll c_{\mathrm{lin}}\), age-structure substantially affects invasion speed.

\begin{figure}[htb]
  \centering
  \includegraphics[width=\textwidth]{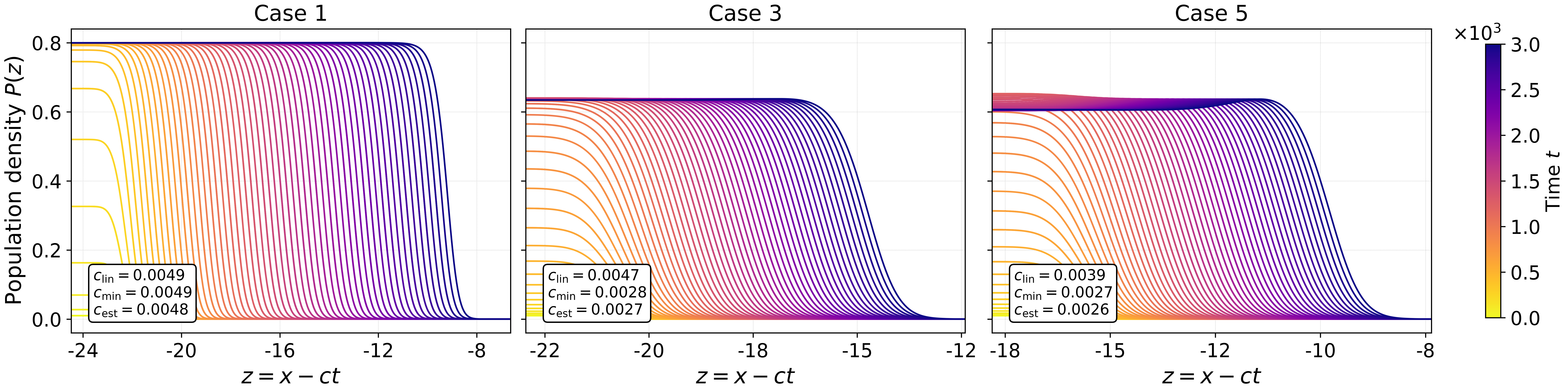}
  \caption{Travelling wave profiles for different model cases. Each panel shows the evolution of the population density \( P(z) \) as a function of the travelling wave coordinate \( z = x - ct \). The numerically estimated wave speed \( c_\mathrm{est} \) approximates the theoretical minimal wave speed \( c_{\min} \). The wave speed predicted by the linear theory \( c_\mathrm{lin} \) is also provided. Parameters: \(\alpha = 0.01\), \(\mu = 0.005\), \(\kappa = 3 \times 10^{-4}\); \(\beta\) satisfies \(\beta = r \times (\text{invasion condition})\) with \(r=5\) and invasion conditions from Table~\ref{tab:comparison-Pc}, which places each case at the same relative distance above its extinction threshold. Case 5 additionally requires \(\gamma < \mu \alpha e\) to ensure non-negative death rates. We set \(\gamma = 1 \times 10^{-4}\).}
  \label{fig:Fig5}
\end{figure}

\begin{figure}[htb]
  \centering
  \includegraphics[width=\textwidth]{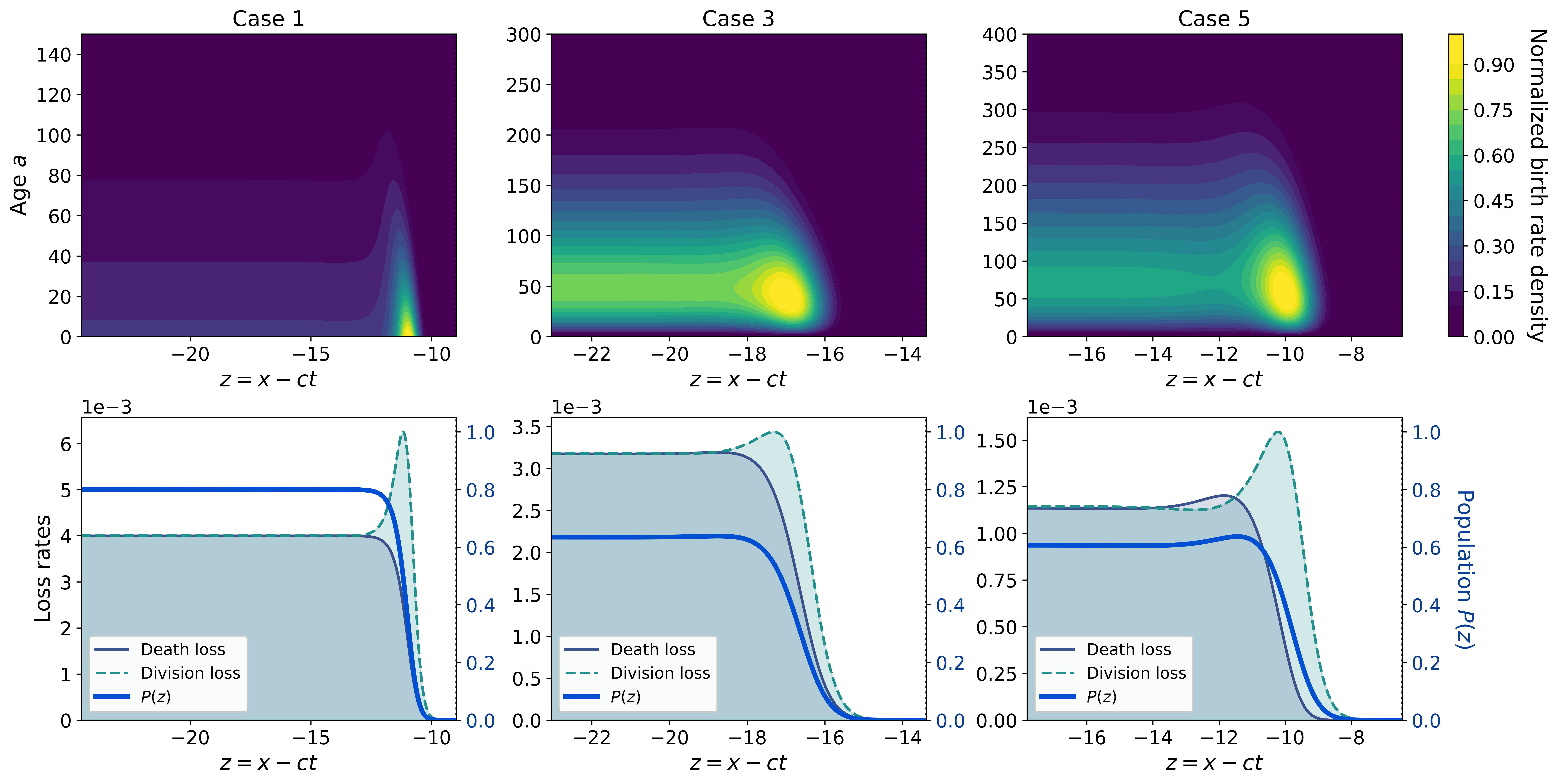}
  \caption{Population structure and birth-death dynamics for different division kernels. Top row: normalised division rate density \(\beta(a,z)u(a,z)\) in the travelling wave frame, showing the spatial distribution of new cells across age classes. Bottom row: death and division loss rates (left axis) alongside total population density \(P(z)\) (right axis), illustrating the balance between losses (cell death and division) and population renewal (birth influx via daughter cells). Parameters are as in Figure~\ref{fig:Fig5}.}
  \label{fig:Fig6}
\end{figure}

\section{Discussion}
We developed an analytical framework for age-structured cell populations that separates division and death processes explicitly and reduces the governing transport-renewal PDE to a tractable integro--differential system. Using this reduction, we derived explicit expressions for the SSAD, the CCTD, and the invasion speeds, and achieved analytical tractability previously limited to homogeneous populations. A fundamental duality emerges: the condition for population persistence is identical to the condition for spatial invasion. If this condition fails, the population does not generate a front and declines to extinction. This correspondence holds across all model cases. Moreover, when the division and death rates are age-independent, the system reduces to a logistic or \FKPP{} equation; when the rates depend explicitly on age, the same framework yields rigorous bounds, including the survival threshold \(P_c\), the wave-speed upper bound \(c_\mathrm{lin}\), and the exact minimal speed \(c_{\min}\). The difference \(c_\mathrm{lin} - c_{\min}\) quantifies how the assumption of a uniform division rate across the cell population leads to an overestimation of the true invasion speed. In modelling applications such as wound healing or tumour invasion, models that neglect age-dependent division rates systematically overestimate the rate of spatial advance. 

\textbf{Age structure changes the rules of population persistence.} 
We first examined the growth dynamics of populations without spatial structure. The framework yields explicit survival conditions (Table~\ref{tab:comparison-Pc}) confirmed by numerical simulations (Fig.~\ref{fig:Fig1}). These conditions follow a systematic pattern. For age-independent dynamics (Case 1), population persistence requires only that the division rate exceeds the death rate (\(\beta > \mu\)). When the division rate varies with age, stricter criteria emerge. An early-division bias, where division rate decays exponentially at rate \(\alpha\) (Case 2), requires \(2\beta > \mu + \alpha\). A maturation delay, where division rate peaks at age \(\sim 1/\alpha\) (Case 3), requires \(2\beta > (\mu + \alpha)^2\). Cases 4--5 represent a qualitatively different regime. When the death rate is density-dependent (\(\mu(a,P) = \mu P\)) with a maturation delay in the division rate, the survival condition simplifies to \(2\beta > \alpha^2\), which depends solely on the division rate. In these cases the division mechanism sets the survival condition, whereas the death rate determines only the equilibrium population size (see \(P_c\) in Table~~\ref{tab:comparison-Pc}).

\textbf{Age structure at invasion fronts slows spatial spread.}
Invasion speed depends on the division rate at low density. Case 1 illustrates the classical scenario where division rates are age-independent, so newborn cells at the leading edge divide at the same rate as older cells, and the wave propagates at the \FKPP{} speed (\(c_{\min} = c_{\mathrm{lin}}\)). In Cases 3 and 5, the leading edge is dominated by newborn cells with low division rates, whereas peak division rate occurs behind the front. This maturation delay explains why \(c_{\min} < c_{\mathrm{lin}}\) in these cases. Cell populations that can divide soon after division (Cases 1--2) invade faster than those requiring maturation before division (Cases 3--5). Far behind the leading edge in all cases, the population reaches equilibrium where the division and death rates balance (Figure~\ref{fig:Fig6}, bottom row), but this region does not drive front propagation.

\textbf{Explicit distributions link mechanism to data.}
The framework yields explicit expressions for the SSAD and the CCTD (Tables~\ref{tab:comparison-Fa},~\ref{tab:comparison-CCTD}). These distributions follow from the same division and death rates but emphasise distinct features of the population. The SSAD describes the age profile of the entire population; the CCTD describes the timing of first divisions. Age-independent division rates produce exponential CCTDs. Gamma-type division rates produce bell-shaped CCTDs that peak at intermediate ages, matching the maturation delays observed in lineage data \cite{smith_1973,leander_2014}. The mean division age and the mean population age can be obtained from the CCTD and the SSAD, respectively. The mean division age systematically lies below the mean population age (Fig.~\ref{fig:Fig3}~h) because cell divisions concentrate among younger cells, whereas the SSAD retains older cells that have neither died nor divided.

Experimental techniques such as fluorescent cell-cycle reporters \cite{sakaue_2008} identify the current phase of an individual cell (G1/S/G2/M) but not the duration of that phase \cite{eastman_2020}. Lineage tracing quantifies cell-cycle time yet requires long-term tracking. The analytical CCTD provides a direct route from data to population-level parameters \((\beta, \alpha, \mu)\), as illustrated in Fig.~\ref{fig:Fig5} for drug treatments where fitted values fall near the theoretical survival thresholds. The same parameters determine the SSAD and predict persistence and invasion speed without live-cell imaging or simulation of the full PDE system. Earlier work established equivalence between multi-stage stochastic models and age-structured PDEs \cite{yates_2017, kynaston_2022}. Our explicit formulae extend that foundation and permit direct parameter inference from standard experimental measurements.

\medskip
Our framework aggregates the cell-cycle phases into a single age variable, which provides analytical tractability but has limitations. It cannot resolve phase-specific regulation (e.g., G1 versus S/G2/M delays) or model heritable differences in division timing---the `two-clock' problem in which biological and chronological age diverge \cite{ocal_2025}. Cells that cycle rapidly may be biologically older despite less chronological time elapsed, and this disparity affects population dynamics in ways our framework does not yet capture. Future work will extend the framework in two directions. First, multi-stage models could test whether the persistence and invasion principles persist under finer cell-cycle resolution. Second, allowing the rate of biological ageing to vary could reveal how internal clocks, distinct from chronological time, modulate population dynamics.

\medskip
\noindent \textbf{Declarations}

\medskip
\noindent \textbf{Authorship and contributorship.}
All authors have made substantial intellectual contributions to the study conception, execution, and design of the work. All authors have read and approved the final manuscript. Contributions were as follows: Conceptualisation: St\'ephanie M. C. Abo, Ruth E. Baker; Methodology: St\'ephanie M. C. Abo, Ruth E. Baker; Software: St\'ephanie M. C. Abo; Formal analysis and investigation: St\'ephanie M. C. Abo; Writing -- original draft preparation: St\'ephanie M. C. Abo; Writing -- review and editing: St\'ephanie M. C. Abo, Ruth E. Baker; Supervision: Ruth E. Baker; Funding acquisition: Ruth E. Baker. 

\medskip
\noindent \textbf{Funding.}
Ruth E. Baker is supported by a grant from the Simons Foundation (MP-SIP-00001828).

\medskip
\noindent \textbf{Acknowledgements.}
We thank Christian A. Yates for sharing unpublished cell-cycle data originally collected by Richard L. Mort and Matthew J. Ford; Rachel N. Leander for providing access to the DMSO, CHX, and ERL datasets available at \href{https://github.com/rnleander/DDT_cell_cycle}{https://github.com/rnleander/DDT\_cell\_cycle}. We also thank Mike B. Giles for helpful discussions on the formulation and implementation of the numerical scheme.

\medskip
\noindent For the purpose of open access, the authors have applied a CC BY public copyright licence to any author accepted manuscript arising from this submission.

\medskip
\noindent \textbf{Conflicts of interest.}
The authors declare there are no conflicts of interest.

\medskip
\noindent \textbf{Consent for publication} All the authors approved the final version of the manuscript.

\medskip
\noindent\textbf{Code availability.} 
All code used for the simulations and figure generation is publicly available on GitHub at 
\href{https://github.com/Stephanie-Abo/age-structured-invasion-dynamics.git}{https://github.com/Stephanie-Abo/age-structured-invasion-dynamics.git}. 

\bibliographystyle{unsrt}
\bibliography{references}

\clearpage


\section*{Supplementary material}

# SUPPLEMENTARY MATERIALS

# SURVIVAL AND INVASION DYNAMICS IN CELL POPULATIONS: AN ANALYTICAL FRAMEWORK FOR THRESHOLD BEHAVIOUR IN NONLINEAR AGE-STRUCTURED MODELS

STÉPHANIE M. C. ABO* AND RUTH E. BAKER*

**SM1. Model summary and reduced moment system.** This document provides supplementary details for the age-structured model of proliferating cell populations analysed in the main text. The model tracks the density of cells, $u(a, x, t)$, of age $a$ (time since last division) at position $x$ and time $t$. The dynamics are governed by the transport–renewal equation with diffusion

$$\text{(SM1.1)} \qquad \partial_t u + \partial_a u = \kappa \Delta u - [\mu(a, P) + \beta(a, P)] u,$$

$$\text{(SM1.2)} \qquad u(0, x, t) = 2 \int_0^\infty \beta(a, P) u(a, x, t) \, \mathrm{d}a,$$

where $P(x, t) = \int_0^\infty u(a, x, t) \, \mathrm{d}a$ is the total population density.

We analyse different model cases (Table SM1): constant rates (age-homogeneous), exponential decay (early division), and gamma-type profiles (maturation delay).

These forms are chosen for analytical tractability and biological relevance. Each case is studied in both the non-spatial ($\kappa = 0$) and spatial ($\kappa > 0$) settings.

TABLE SM1
*Model cases studied in both non-spatial and spatial settings*

| Case | $\mu(a, P)$ | $\beta(a, P)$ |
|------|-------------|---------------|
| 1    | $\mu$       | $\beta(1 - P)$ |
| 1b   | $\mu P$     | $\beta(1 - P)$ |
| 2    | $\mu$       | $\beta e^{-\alpha a}(1 - P)$ |
| 3    | $\mu$       | $\beta a e^{-\alpha a}(1 - P)$ |
| 4    | $\mu P$     | $\beta a e^{-\alpha a}(1 - P)$ |
| 5    | $(\mu - \gamma a e^{-\alpha a}) P$ | $\beta a e^{-\alpha a}(1 - P)$ |

**Generalised moment system.** We derive an integro–differential reduction of the age-structured system (SM1.1)–(SM1.2). We introduce age-weighted moments that mirror the form of division rates.

For any fixed decay parameter $\alpha > 0$, the reduced system is formulated in terms of the following quantities,

$$\text{(SM1.3)} \qquad P(x, t) = \int_0^\infty u(a, x, t) \, \mathrm{d}a,$$

$$\text{(SM1.4)} \qquad C(x, t; \alpha) = \int_0^\infty e^{-\alpha a} u(a, x, t) \, \mathrm{d}a,$$

$$\text{(SM1.5)} \qquad D(x, t; \alpha) = \int_0^\infty a e^{-\alpha a} u(a, x, t) \, \mathrm{d}a.$$

---

*Mathematical Institute, University of Oxford, UK (stephanie.abo@maths.ox.ac.uk).

The system $(P, C, D)$ is not closed, and its evolution involves higher moments,

$$(\text{SM1.6}) \qquad I_0(t, x; \alpha) = \int_0^\infty e^{-2\alpha a} u(a, x, t) \, \mathrm{d}a,$$

$$(\text{SM1.7}) \qquad I_1(t, x; \alpha) = \int_0^\infty a e^{-2\alpha a} u(a, x, t) \, \mathrm{d}a,$$

$$(\text{SM1.8}) \qquad I_2(t, x; \alpha) = \int_0^\infty a^2 e^{-2\alpha a} u(a, x, t) \, \mathrm{d}a.$$

The birth term (birth boundary condition) is case-dependent and involves these moments,

$$(\text{SM1.9}) \qquad B(x, t) = u(0, x, t) = \begin{cases} 2\beta(1-P)P, & \text{Case 1,} \\ 2\beta(1-P)C, & \text{Case 2,} \\ 2\beta(1-P)D, & \text{Cases 3–5.} \end{cases}$$

The resulting generalised integro–differential system is

$$(\text{SM1.10}) \qquad \partial_t P = \kappa \Delta P + B - \int_0^\infty [\mu(a, P(x,t)) + \beta(a, P(x,t))] u(a, x, t) \, \mathrm{d}a,$$

$$(\text{SM1.11}) \qquad \partial_t C = \kappa \Delta C + B - \alpha C - \int_0^\infty [\mu(a, P(x,t)) + \beta(a, P(x,t))] e^{-\alpha a} u(a, x, t) \, \mathrm{d}a,$$

$$(\text{SM1.12}) \qquad \partial_t D = \kappa \Delta D + C - \alpha D - \int_0^\infty [\mu(a, P(x,t)) + \beta(a, P(x,t))] a e^{-\alpha a} u(a, x, t) \, \mathrm{d}a,$$

with $B(x,t)$ (SM1.9) and higher moments $I_k(x, t; \alpha) = \int_0^\infty a^k e^{-2\alpha a} u(a, x, t) \mathrm{d}a$. The integral terms are evaluated using the defined moments and the specific forms of $\mu(a, P)$ and $\beta(a, P)$ for each model case. In general, age-homogeneous rates (constant rates) yield exact closure in $P(t)$; exponential decay rates (early reproduction bias) require bounds on $P, C, I_0$; gamma-type rates (maturation delay) require bounds on $P, C, D, I_1, I_2$.

**SM2. Boundedness of age-weighted moments.**

PROPOSITION SM2.1 (Boundedness in the non-spatial case). *Suppose $\kappa = 0$ and $P(0) \le 1$. Then for all $t \ge 0$:*

*(i) $P(t) \le 1$,*
*(ii) The moments $C(t; \alpha)$, $D(t; \alpha)$, and $I_k(t; \alpha)$ remain bounded.*

*Proof.* **(i) Population bound.** The total population satisfies

$$\frac{\mathrm{d}P}{\mathrm{d}t} = B(t) - \int_0^\infty [\mu(a, P) + \beta(a, P)] u(a, t) \mathrm{d}a,$$

where $B(t) = 2(1 - P) \int_0^\infty \beta_0(a) u(a, t) \mathrm{d}a$.

Whenever $P(t) \ge 1$, we have $B(t) \le 0$ whilst the loss term remains non-negative. Thus $\mathrm{d}P/\mathrm{d}t \le 0$ in any region where $P \ge 1$. Since $P(0) \le 1$, it follows that $P(t) \le 1$ for all $t \ge 0$.

**(ii) Moment bounds.** Since $e^{-\alpha a} < 1$ for all $a > 0$,

$$C(t; \alpha) = \int_0^\infty e^{-\alpha a} u(a, t) \mathrm{d}a < P(t) \le 1.$$

The function $a e^{-\alpha a}$ attains its maximum $1/(\alpha e)$ at $a = 1/\alpha$, giving

$$D(t; \alpha) = \int_0^\infty a e^{-\alpha a} u(a, t) \mathrm{d}a \le \frac{P(t)}{\alpha e} \le \frac{1}{\alpha e}.$$

Similarly, $a^k e^{-2\alpha a}$ attains its maximum $(k/(2\alpha e))^k$ at $a = k/(2\alpha)$, hence

$$I_k(t; \alpha) = \int_0^\infty a^k e^{-2\alpha a} u(a, t) \mathrm{d}a \le \left(\frac{k}{2\alpha e}\right)^k P(t) \le \left(\frac{k}{2\alpha e}\right)^k.$$

These bounds ensure all moments remain finite. ∎

**SM3. Moment-based derivations of survival thresholds.**

**SM3.1. Fully closed population dynamics.** When division and death rates depend only on total density, the population is age-homogeneous. All cells share the same division and death rates, and the moment system closes exactly, as established by Lemma 1 in the main text.

*Case 1.* Consider the rates

$$\mu(a, P) = \mu, \qquad \beta(a, P) = \beta(1 - P), \tag{SM3.1}$$

with constants $\mu, \beta > 0$. This models a population where division is subject to contact inhibition, while cell death occurs at a constant rate independent of local conditions. The total population $P(t)$ satisfies the logistic equation,

$$P'(t) = (\beta - \mu) P(t) \left(1 - \frac{P(t)}{K}\right), \quad \text{with } K = \frac{\beta - \mu}{\beta}.$$

Population persistence requires $\beta > \mu$. When this holds, the unique positive steady state defines the critical threshold

$$P_c = 1 - \frac{\mu}{\beta}. \tag{SM3.2}$$

**SM3.2. Exponential decay in division potential.** This models a division rate that declines exponentially with age, $\beta(a, P) \propto e^{-\alpha a}$. This is a reasonable approximation if the population reproduces at sufficiently young ages.

*Case 2.* Consider the rates

$$\mu(a, P) = \mu, \qquad \beta(a, P) = \beta e^{-\alpha a}(1 - P), \tag{SM3.3}$$

with constants $\mu, \beta, \alpha > 0$. Division potential is highest immediately after birth and decays exponentially thereafter. The system (SM1.1)–(SM1.2) reduces to

$$P'(t) = -\mu P(t) + \beta(1 - P(t))C(t),$$
$$C'(t) = -(\mu + \alpha)C(t) + \beta(1 - P(t))(2C(t) - I_0(t)),$$

where $C(t) = \int_0^\infty e^{-\alpha a} u \, da$ and $I_0(t) = \int_0^\infty e^{-2\alpha a} u \, da$.

At steady state,

$$\bar{C} = \frac{\mu \bar{P}}{\beta(1 - \bar{P})}, \quad \bar{I}_0 = 2\bar{C} - \frac{(\mu + \alpha) \bar{C}}{\beta(1 - \bar{P})}, \quad \text{for } \bar{P} \neq 1.$$

The non-negativity constraint $\bar{I}_0 \geq 0$ implies

$$2\bar{C} - \frac{(\mu + \alpha) \bar{C}}{\beta(1 - \bar{P})} \geq 0.$$

Assuming $\bar{C} \neq 0$ and $\bar{P} \neq 1$, this simplifies to $2\beta(1 - \bar{P}) \geq \mu + \alpha$. Therefore, population persistence requires $2\beta > \mu + \alpha$, and any positive steady state satisfies $\bar{P} < P_c$, with

$$P_c = 1 - \frac{\mu + \alpha}{2\beta}. \tag{SM3.4}$$

If $2\beta \leq \mu + \alpha$, only the extinction equilibrium exists.

**SM3.3. Gamma-type age dependence.** The division rate follows a gamma-type profile $\beta(a, P) \propto a e^{-\alpha a}$, capturing maturation delays before division. We consider three cases: constant mortality (Case 3), density-dependent mortality (Case 4), and age- and density-dependent mortality (Case 5).

**Case 3.** Consider the rates

$$\mu(a,P) = \mu, \qquad \beta(a,P) = \beta a e^{-\alpha a}(1-P), \tag{SM3.5}$$

with constants $\mu, \beta, \alpha > 0$. The system (SM1.1)–(SM1.2) reduces to

$$\begin{aligned} P'(t) &= -\mu P(t) + \beta(1-P(t))\,D(t), \\ C'(t) &= -(\mu+\alpha)\,C(t) + \beta(1-P(t))\,(2D(t) - I_1(t)), \\ D'(t) &= -(\mu+\alpha)\,D(t) + C(t) - \beta(1-P(t))\,I_2(t), \end{aligned} \tag{SM3.6}$$

with

$$C(t;\alpha) = \int_0^\infty e^{-\alpha a} u(a,t)\,\mathrm{d}a, \qquad D(t;\alpha) = \int_0^\infty a e^{-\alpha a} u(a,t)\,\mathrm{d}a,$$

and higher moments

$$I_1(t;\alpha) = \int_0^\infty a e^{-2\alpha a} u(a,t)\,\mathrm{d}a, \qquad I_2(t;\alpha) = \int_0^\infty a^2 e^{-2\alpha a} u(a,t)\,\mathrm{d}a.$$

At steady state, assuming $\bar{D} \neq 0$, we obtain

$$\bar{D} = \frac{\mu\bar{P}}{\beta(1-\bar{P})}, \qquad \bar{I}_1 = 2\bar{D} - \frac{(\mu+\alpha)\bar{C}}{\beta(1-\bar{P})}, \qquad \bar{I}_2 = \frac{\bar{C} - (\mu+\alpha)\bar{D}}{\beta(1-\bar{P})}, \quad \text{for } \bar{P} \neq 1.$$

Given $\bar{I}_1 \geq 0$, we have

$$\bar{C} \leq \frac{2\mu\bar{P}}{\mu+\alpha}.$$

Substituting into $\bar{I}_2 \geq 0$ yields

$$\bar{C}\beta(1-\bar{P}) - (\mu+\alpha)\mu\bar{P} \geq 0 \quad \Longrightarrow \quad \bar{P} \leq \frac{\bar{C}\beta(1-\bar{P})}{(\mu+\alpha)\mu} \leq \frac{2\mu\bar{P}\beta(1-\bar{P})}{\mu(\mu+\alpha)^2}.$$

Assuming $\bar{P} \neq 0$, this gives

$$\bar{P} \leq 1 - \frac{(\mu+\alpha)^2}{2\beta}.$$

For $\bar{P}$ to be positive, we require $2\beta > (\mu+\alpha)^2$. When this condition holds,

$$P_c = 1 - \frac{(\mu+\alpha)^2}{2\beta} \tag{SM3.7}$$

provides an upper bound on steady-state populations.

**Case 4.** Consider that death rate increases proportionally with total population,

$$\mu(a,P) = \mu P, \qquad \beta(a,P) = \beta a e^{-\alpha a}(1-P), \tag{SM3.8}$$

with constants $\mu, \beta, \alpha > 0$. The moment hierarchy includes $P, C, D, I_1$, and $I_2$,

$$\begin{aligned} P'(t) &= -\mu P(t)^2 + \beta(1-P(t))D(t), \\ C'(t) &= -(\mu P(t)+\alpha)C(t) + \beta(1-P(t))\,(2D(t) - I_1(t)), \\ D'(t) &= -(\mu P(t)+\alpha)D(t) + C(t) - \beta(1-P(t))I_2(t). \end{aligned} \tag{SM3.9}$$

At steady state,

$$\bar{D} = \frac{\mu\bar{P}^2}{\beta(1-\bar{P})}, \quad \bar{I}_1 = 2\bar{D} - \frac{\bar{C}(\alpha+\mu\bar{P})}{\beta(1-\bar{P})} \quad \bar{I}_2 = \frac{\bar{C} - \bar{D}(\alpha+\mu\bar{P})}{\beta(1-\bar{P})}, \text{ for } \bar{P} \neq 1.$$

Given $\bar{I}_1 \geq 0$,

$$\bar{C} \leq \frac{2\bar{D}\beta(1-\bar{P})}{\mu\bar{P}+\alpha}.$$

Similarly, $\bar{I}_2 \geq 0$ implies
$$\bar{D}(\alpha + \mu \bar{P}) \leq \bar{C}.$$

Combining these inequalities gives upper and lower bounds on $\bar{C}$, which lead to

(SM3.10) $$(\mu \bar{P} + \alpha)^2 \leq 2\beta(1 - \bar{P}).$$

The above inequality simplifies to

(SM3.11) $$\bar{P} \leq \frac{-(\alpha\mu + \beta) + \sqrt{(\alpha\mu + \beta)^2 - \mu^2(\alpha^2 - 2\beta)}}{\mu^2}, \quad \text{with} \quad 2\beta > \alpha^2.$$

Population survival thus requires $2\beta > \alpha^2$. When this holds,

(SM3.12) $$P_c = \frac{-(\alpha\mu + \beta) + \sqrt{(\alpha\mu + \beta)^2 - \mu^2(\alpha^2 - 2\beta)}}{\mu^2}.$$

**Case 5.** Consider that both death and division rates depend explicitly on cell age,

(SM3.13) $$\mu(a, P) = (\mu - \gamma \, a e^{-\alpha a}) \, P, \qquad \beta(a, P) = \beta \, a e^{-\alpha a}(1 - P),$$

with $\mu, \gamma, \beta, \alpha > 0$. We require $\mu > \gamma/(\alpha e)$ to ensure positive death rates. This models a scenario where the death rate is initially elevated for newly divided cells, decreases during a maturation phase, and rises again due to senescence. The moment hierarchy becomes

(SM3.14)
$$\begin{aligned}
P'(t) &= -\mu P(t)^2 + \gamma P(t) D(t) + \beta(1 - P(t))D(t), \\
C'(t) &= -(\mu P(t) + \alpha)C(t) + \gamma P(t) I_1(t) + \beta(1 - P(t)) \left(2D(t) - I_1(t)\right), \\
D'(t) &= -(\mu P(t) + \alpha)D(t) + C(t) + \gamma P(t) I_2(t) - \beta(1 - P(t)) I_2(t).
\end{aligned}$$

At steady state with $\bar{P} \neq 1$ and $\gamma \bar{P} \neq \beta(1 - \bar{P})$ (ensuring the moment expressions are well-defined),

$$\bar{D} = \frac{\mu \bar{P}^2}{\gamma \bar{P} + \beta(1 - \bar{P})}, \quad \bar{I}_1 = \frac{(\mu \bar{P} + \alpha)\bar{C} - 2\beta(1 - \bar{P})\bar{D}}{\gamma \bar{P} - \beta(1 - \bar{P})}, \quad \bar{I}_2 = \frac{(\mu \bar{P} + \alpha)\bar{D} - \bar{C}}{\gamma \bar{P} - \beta(1 - \bar{P})}.$$

For non-negative moments $\bar{I}_1 \geq 0$ and $\bar{I}_2 \geq 0$, both numerator and denominator must have the same sign. The condition $\gamma \bar{P} \neq \beta(1 - \bar{P})$ ensures the denominators are non-zero.

We analyse first the case where both the numerator and denominator are negative in the expressions for $\bar{I}_1$ and $\bar{I}_2$. From $\bar{I}_1 \geq 0$ we obtain
$$\bar{C} \leq \frac{2\beta(1 - \bar{P})}{\mu \bar{P} + \alpha} \bar{D}.$$

Then from $\bar{I}_2 \geq 0$ we obtain
$$\bar{C} \geq (\mu \bar{P} + \alpha)\bar{D}.$$

Combining these inequalities gives the constraint

(SM3.15) $$(\mu \bar{P} + \alpha)^2 \leq 2\beta(1 - \bar{P}).$$

The function $f(P) = (\mu P + \alpha)^2 + 2\beta P - 2\beta$ is strictly increasing for $P \geq 0$, since

$$f'(P) = 2\mu(\mu P + \alpha) + 2\beta > 0.$$

Therefore, the inequality $f(\bar{P}) \leq 0$ is equivalent to $\bar{P} \leq P_c$, where $P_c$ is the unique positive root of $f(P) = 0$,

(SM3.16) $$P_c = \frac{-(\alpha\mu + \beta) + \sqrt{(\alpha\mu + \beta)^2 - \mu^2(\alpha^2 - 2\beta)}}{\mu^2}.$$

The condition $2\beta > \alpha^2$ ensures the discriminant is positive and $P_c > 0$. From the denominator we have $\bar{P} < \beta/(\gamma + \beta)$ (strict inequality since $\gamma\bar{P} \neq \beta(1 - \bar{P})$). The upper bound on the positive steady state becomes
$$\bar{P} < \min\left(P_c,\ \beta/(\gamma + \beta)\right).$$

When both numerator and denominator are positive, the analogous derivation yields $(\mu\bar{P} + \alpha)^2 \geq 2\beta(1 - \bar{P})$, which combined with $\bar{P} > \beta/(\gamma + \beta)$ provides a lower bound on feasible steady states. However, for survival threshold analysis, the first case provides the relevant upper bound. The survival condition $2\beta > \alpha^2$ ensures positive steady states exist with population bounded by the constraints above.

The parameter $\gamma$ determines how strongly the death rate varies with cell age. In the limit $\gamma \to 0$, the model reduces to Case 4 with purely density-dependent death, and $\beta/(\gamma + \beta) \to 1$, leaving $P_c$ as the sole bound. As $\gamma$ increases, the constraint $\mu > \gamma/(\alpha e)$ forces proportionally larger baseline death rates $\mu$, whilst $\beta/(\gamma + \beta) \to 0$ drives the equilibrium toward extinction. Although positive steady states persist when $2\beta > \alpha^2$, large $\gamma$ yields populations barely above zero. For typical biological scenarios, $\gamma$ remains small relative to $\beta$, ensuring $P_c$ provides the relevant upper bound.

**SM4. Minimal wave speed derivations.** We present the full derivation for Case 5, which is the most general. Case 1 is derived in the main text. Cases 2–4 follow by simplification (replacing $\mu$ with $\mu P$, setting $\gamma = 0$, or removing age-dependent terms).

Consider the following division and death rates,

(SM4.1) $$\mu(a, P) = (\mu - \gamma a e^{-\alpha a})P, \qquad \beta(a, P) = \beta a e^{-\alpha a}(1 - P).$$

The survival condition from the non-spatial analysis is $2\beta > \alpha^2$, which we assume holds. We further require $\mu > \gamma/(\alpha e)$ to ensure $\mu(a, P) > 0$ for all $a \geq 0$.

We derive the minimal speed $c_{\min}$ via the dispersion relation arising from the full age-structured PDE. Starting from

(SM4.2) $$\frac{\partial u}{\partial t} + \frac{\partial u}{\partial a} = \kappa \frac{\partial^2 u}{\partial x^2} - [\mu(a, P) + \beta(a, P)]u,$$

with boundary condition

(SM4.3) $$u(0, x, t) = 2\int_0^\infty \beta(a, P(x,t))u(a, x, t)\,\mathrm{d}a,$$

the travelling-wave ansatz $u(a, x, t) = \phi(a, z)$, $z = x - ct$ transforms the system into

(SM4.4) $$-c\frac{\partial \phi}{\partial z} + \frac{\partial \phi}{\partial a} = \kappa \frac{\partial^2 \phi}{\partial z^2} - [(\mu - \gamma a e^{-\alpha a})U(z) + \beta a e^{-\alpha a}(1 - U(z))]\phi,$$

where $U(z) = \int_0^\infty \phi(a, z)\,\mathrm{d}a$.

Linearising at the leading edge where $U \to 0$, the death term $(\mu - \gamma a e^{-\alpha a})U \to 0$ and $(1 - U) \to 1$, giving

(SM4.5) $$-c\frac{\partial \phi}{\partial z} + \frac{\partial \phi}{\partial a} = \kappa \frac{\partial^2 \phi}{\partial z^2} - \beta a e^{-\alpha a}\phi + O(U).$$

The boundary condition linearises to

(SM4.6) $$\phi(0, z) = 2\beta \int_0^\infty a e^{-\alpha a}\phi(a, z)\,\mathrm{d}a + O(U).$$

We assume separability $\phi(a, z) = A(a)e^{-\lambda z}$ and solve the resulting age-profile equation. Here $\lambda$ is the rate of spatial decay. This yields the dispersion relation

(SM4.7) $$1 = 2\beta \int_0^\infty a e^{-\alpha a} \exp\left[-(c\lambda - \kappa\lambda^2)a - \frac{\beta}{\alpha^2}(1 - e^{-\alpha a}(1 + \alpha a))\right]\mathrm{d}a.$$

The minimal speed $c_{\min}$ is the smallest $c > 0$ for which there exists $\lambda > 0$ satisfying equation (SM4.7).

We now derive the upper bound $c_{\text{lin}}$ from the moment system. Using the reduction framework of Section 2 in the main text and the travelling-wave ansatz $z = x - ct$, we have

(SM4.8)
$$\begin{aligned}
-cU' &= \kappa U'' - \mu U^2 + \gamma UW + \beta(1-U)W, \\
-cV' &= \kappa V'' - (\mu U + \alpha)V + \gamma U \tilde{I}_1 + \beta(1-U)\left(2W - \tilde{I}_1\right), \\
-cW' &= \kappa W'' - (\mu U + \alpha)W + V + \gamma U \tilde{I}_2 - \beta(1-U)\tilde{I}_2.
\end{aligned}$$

where $P(x,t) = U(z)$, $C(x,t) = V(z)$, $D(x,t) = W(z)$, and $I_k(x,t) = \tilde{I}(z)$. Converting to a first-order system with $Y_1 = U$, $Y_2 = U'$, $Y_3 = V$, $Y_4 = V'$, $Y_5 = W$, $Y_6 = W'$, the linearisation at $(Y_1, \ldots, Y_6) = (0, \ldots, 0)$ yields

$$\mathbf{Y}' = \mathbf{J}\mathbf{Y}, \quad \mathbf{J} = \begin{pmatrix} 0 & 1 & 0 & 0 & 0 & 0 \\ 0 & -\frac{c}{\kappa} & 0 & 0 & -\frac{\beta}{\kappa} & 0 \\ 0 & 0 & 0 & 1 & 0 & 0 \\ 0 & 0 & \frac{\alpha}{\kappa} & -\frac{c}{\kappa} & -\frac{2\beta}{\kappa} & 0 \\ 0 & 0 & 0 & 0 & 0 & 1 \\ 0 & 0 & -\frac{1}{\kappa} & 0 & \frac{\alpha}{\kappa} & -\frac{c}{\kappa} \end{pmatrix}.$$

The eigenvalues are

$$\lambda_1 = 0, \quad \lambda_2 = -\frac{c}{\kappa},$$

$$\lambda_{3,4} = \frac{-c \pm \sqrt{c^2 + 4(\alpha - \sqrt{2\beta})\kappa}}{2\kappa},$$

$$\lambda_{5,6} = \frac{-c \pm \sqrt{c^2 + 4(\alpha + \sqrt{2\beta})\kappa}}{2\kappa}.$$

We seek biologically realistic solutions, and hence the eigenvalues must be real. The critical constraint comes from $\lambda_3$, which has positive real part unless $c^2 + 4(\alpha - \sqrt{2\beta})\kappa \leq 0$. This requires

$$c \geq 2\sqrt{\kappa(\sqrt{2\beta} - \alpha)},$$

yielding the upper bound $c_{\text{lin}} = 2\sqrt{\kappa(\sqrt{2\beta} - \alpha)}$, valid when $\sqrt{2\beta} > \alpha$.

The inequality $c_{\min} \leq c_{\text{lin}}$ reflects the fact that the moment system linearisation assumes all age classes vanish uniformly at the wave front, neglecting the natural age structure where younger cells dominate the leading edge. The dispersion relation preserves this age-dependent decay structure at the leading edge, yielding a sharper (and typically slower) minimal speed bound. Both approaches give the same threshold condition for invasion: $\sqrt{2\beta} > \alpha$, equivalent to the non-spatial survival condition $2\beta > \alpha^2$.

**SM5. Numerical methods.** The system (SM1.1)–(SM1.2) combines three distinct processes: advective transport in the age variable, diffusive spread in space, and a non-local renewal condition at $a = 0$ that links all ages through cell division. Because these processes act on different variables and timescales, advancing them simultaneously introduces numerical errors. The diffusion operator acts locally in space, coupling only neighbouring points, whereas the renewal term introduces a global coupling across the age domain. Standard reaction–diffusion schemes fail to maintain the non-negativity of $u(a,x,t)$ and to balance the total renewal flux. Additional numerical challenges arise from the unbounded age domain and from stiffness when the combined loss rate $\mu(a,P) + \beta(a,P)$ changes rapidly with age or density.

We address these challenges with an operator-splitting scheme that advances age transport, reaction, and diffusion sequentially, combined with a predictor–corrector integration that evaluates rates at both the current and predicted states to achieve second-order accuracy. The renewal boundary condition is computed from the corrected interior solution using averaged division rates. This ensures consistent coupling between the non-local birth integral and the interior dynamics while preserving mass balance and non-negativity throughout the domain.

**Discretisation.** The computational domain is $x \in [-L, L]$, $a \in [0, a_{\max}]$, and $t \in [0, T]$. We discretise the solution $u(a, x, t)$ on a uniform grid: $a_j = j\,\Delta a$ for age, $x_i = i\,\Delta x$ for space, and $t_n = n\,\Delta t$ for time. Here $u_{j,i}^n$ denotes the numerical approximation to $u(a_j, x_i, t_n)$.

For the age variable, the discrete transport operator is constructed along the exact characteristics of the hyperbolic part, $\partial_t u + \partial_a u = 0$, whose solutions satisfy $u(a, t) = u(a - t, 0)$. Setting $\Delta t = \Delta a$ ensures that

$$u_{j+1,i}^{n+1} = u_{j,i}^n,$$

meaning that each cohort of cells ages exactly one bin per time step. This choice guarantees positivity and eliminates artificial smoothing in the age coordinate. It also couples the time step to the age grid, a mild but necessary constraint for accuracy and stability. Stability is maintained by ensuring that $\Delta a = \Delta t$ is sufficiently small relative to the reaction timescales.

For spatial diffusion, we employ the Crank–Nicolson scheme, treating the diffusion operator implicitly via a centred three-point stencil. At each time step, the update for age slice $j$ satisfies

$$\frac{u_{j,i}^{n+1} - u_{j,i}^n}{\Delta t} = \frac{\kappa}{2} \frac{u_{j,i-1}^{n+1} - 2u_{j,i}^{n+1} + u_{j,i+1}^{n+1}}{(\Delta x)^2} + \frac{\kappa}{2} \frac{u_{j,i-1}^n - 2u_{j,i}^n + u_{j,i+1}^n}{(\Delta x)^2},$$

which is second-order accurate in both space and time for the diffusive subproblem and unconditionally stable. We impose zero-flux boundary conditions at $x = \pm L$: at the left boundary, $u_{-1,j} = u_{1,j}$; at the right boundary, $u_{n_x+1,j} = u_{n_x-1,j}$. Thus $\partial_x u = 0$ at the edges.

**Operator splitting and nonlinear coupling.** Each full step $t_n \to t_{n+1}$ proceeds through an operator-split sequence (age transport, reaction, diffusion) embedded within a predictor–corrector framework that makes the scheme second order in time.

In the predictor step, we compute rates $\mu(a, P^n)$ and $\beta(a, P^n)$ from the current population density $P^n = \int_0^{a_{\max}} u^n(a, x)\,\mathrm{d}a$, apply the explicit reaction update, and then advance diffusion via Crank–Nicolson to obtain the predicted solution $u^{\mathrm{pred}}$. We recompute $P^{\mathrm{pred}} = \int_0^{a_{\max}} u^{\mathrm{pred}}(a, x)\,\mathrm{d}a$ and evaluate rates at this predicted state. The corrector step applies time-averaged rates $\bar{\mu} = \frac{1}{2}[\mu(a, P^n) + \mu(a, P^{\mathrm{pred}})]$ and $\bar{\beta} = \frac{1}{2}[\beta(a, P^n) + \beta(a, P^{\mathrm{pred}})]$ to the original $u^n$, again sequencing reaction then diffusion, yielding the corrected solution $u^{\mathrm{corr}}$. This split predictor–corrector scheme achieves $\mathcal{O}(\Delta t^2)$ temporal accuracy. The renewal boundary condition at $a = 0$ is computed from the corrected interior solution using the averaged division rates. This ensures consistent boundary–interior coupling.

**Age transport and renewal.** The age variable evolves along the exact characteristics of the transport equation $\partial_t u + \partial_a u = 0$. With $\Delta t = \Delta a$, each cohort advances deterministically via $u_{j+1,i}^{n+1} = u_{j,i}^{n+1}$ for $j \geq 1$. The boundary at $a = 0$ is determined by the discrete renewal condition

$$u_{0,i}^{n+1} = 2 \sum_{j=0}^{N_a - 1} \bar{\beta}(a_j, P_i)\, u_{j,i}^{\mathrm{corr}}\, w_j,$$

where $w_j$ are Simpson quadrature weights, $\bar{\beta}$ denotes the time-averaged division rate from the predictor–corrector scheme, and $u^{\mathrm{corr}}$ is the corrected interior solution. This maintains the integral balance and preserves non-negativity provided $\beta \geq 0$.

The unbounded age domain $[0, \infty)$ is truncated dynamically at $a_{\max}(t)$. At each step we compute the tail fraction

$$\eta(t) = \frac{\int_{-L}^{L} u(a_{\max}, x, t)\,\mathrm{d}x\,\Delta a}{\int_0^{a_{\max}} \int_{-L}^{L} u(a, x, t)\,\mathrm{d}x\,\mathrm{d}a}.$$

If $\eta(t) > \eta_{\mathrm{tol}} = 10^{-6}$, we extend the grid by one age bin; otherwise, the oldest cohort is removed. This adaptive procedure bounds the relative truncation error in total mass by $\mathcal{O}(\eta_{\mathrm{tol}})$ while maintaining computational cost linear in the number of active age classes.

**SM5.1. Numerical validation.** We validate the scheme through three complementary tests. First, mass conservation: in non-spatial simulations across all five model cases, the relative mass drift remains below $10^{-6}$ after transients (Figure SM1). Second, travelling wave speeds: numerical estimates $c_{\mathrm{est}}$ converge

to theoretical predictions $c_{\min}$ across Cases 1–5, confirming the linearisation analysis and dispersion relations (Figure SM2). Third, grid refinement: $L^2$ errors in both $P(x)$ and $u(a, x)$ exhibit first-order convergence in $\Delta a$ and $\Delta x$ (Figure SM3), consistent with the characteristic marching in the age coordinate, which is formally first-order despite the higher-order treatment of diffusion and reaction terms. Together, these tests confirm the scheme faithfully reproduces the analytical framework whilst maintaining conservation laws and the predicted wave dynamics.

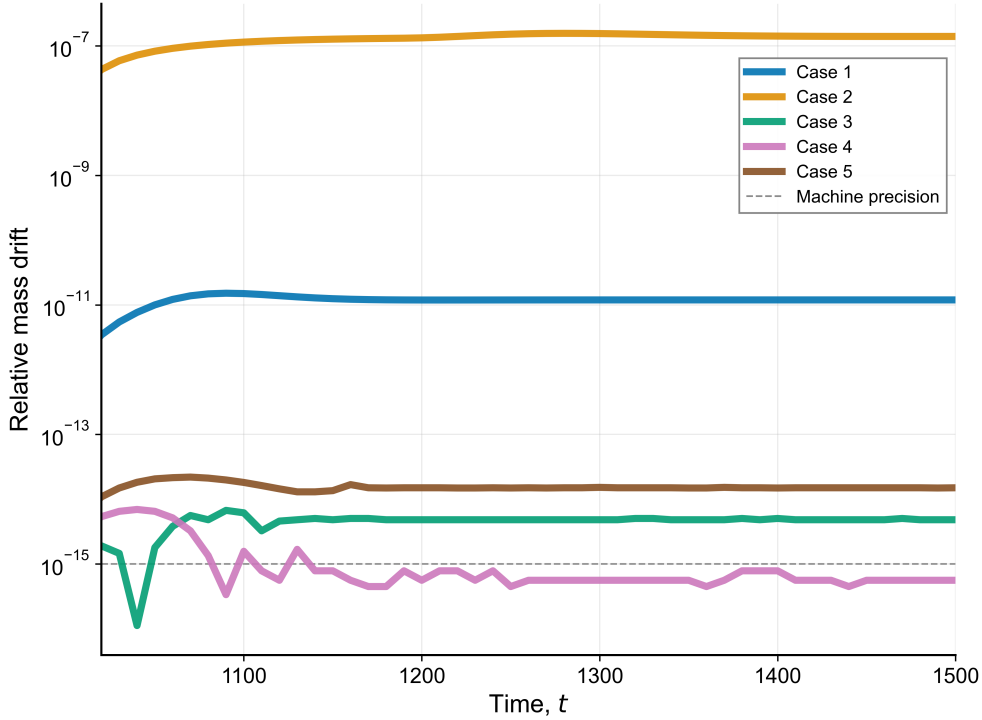

FIG. SM1. *Mass conservation validation for the non-spatial numerical scheme. Relative mass drift over time is shown for all model cases (1–5) on a logarithmic scale. Simulations use parameters $T = 1500$, $\Delta a = 10^{-2}$, $\alpha = 0.01$, $\mu = 0.012$, $\beta = 0.065$, and $\gamma = 10^{-5}$. The parameter $\gamma$ applies only to Case 5. The initial condition is $u(a, 0) = e^{-10a^2}$. All cases conserve total mass, with relative drift remaining below $10^{-6}$.*

**SM6. Parameter estimation from cell-cycle time data.** Parameters for Case 3 $(\beta, \alpha, \mu)$ were estimated from experimental intermitotic time distributions using a profile likelihood approach. The minimum cell-cycle time $a_0$ was fixed at the 10th percentile of each dataset to define the left support of the distribution.

For a grid of $\alpha$ values spanning a plausible range determined by the data quantiles (specifically, $\alpha \in [1/t_{50}, 1/t_{10}]$ where $t_{10}$ and $t_{50}$ are the 10th and 50th percentiles of the shifted data), we optimised $(\beta, \mu)$ by minimising the combined objective function

$$L(\beta, \mu \mid \alpha) = L_{\mathrm{CDF}} + \lambda\, L_{\mathrm{tail}} + P(\beta, \mu, \alpha),$$

where

- $L_{\mathrm{CDF}} = \int_{a_0}^{a_{\max}} [F_{\mathrm{model}}(a) - F_{\mathrm{data}}(a)]^2 \, \mathrm{d}a$ measures the $L^2$ error between the model and empirical cumulative distribution functions (CDFs),
- $L_{\mathrm{tail}} = \frac{1}{3} \sum_i [S_{\mathrm{model}}(a_i) - S_{\mathrm{data}}(a_i)]^2$ penalises deviations in survival probabilities at anchor points $a_i$ corresponding to the 80th percentile, 90th percentile, and maximum observed time,
- $P(\beta, \mu, \alpha) = 10^6 \cdot \max(0, (\mu+\alpha)^2 - 2\beta)^2$ enforces the survival condition $2\beta > (\mu+\alpha)^2$ via a quadratic penalty.

The weight $\lambda = 0.5$ balances CDF accuracy with tail behavior. Results were insensitive to $\lambda \in [0.3, 0.7]$. The triplet $(\beta, \alpha, \mu)$ minimising $L$ across all $\alpha$ values was selected as the final estimate. Optimisation was

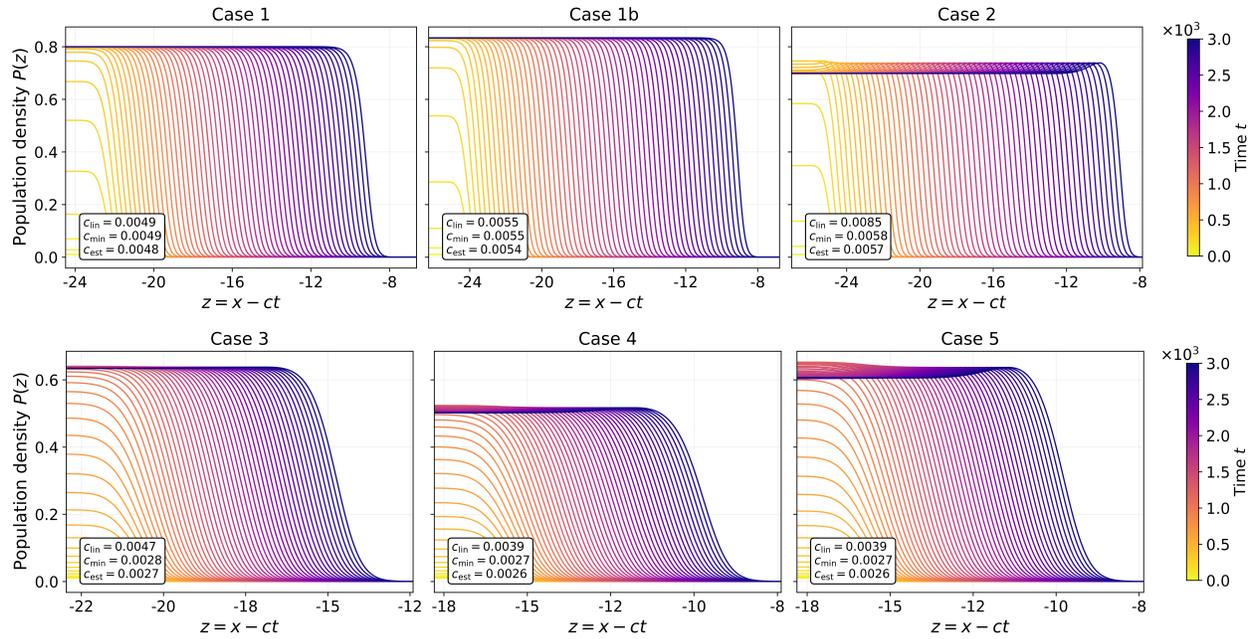

FIG. SM2. *Travelling-wave solutions for Cases1–5. Each panel shows the population density $P(z)$ in the moving frame $z = x - ct$. Numerical solutions converge to stable travelling waves with speeds matching theoretical predictions: for age-independent rates (Cases 1, 1b), $c_{\min} = c_{\text{lin}}$; for models with maturation delay (Cases 2–5), invasion is slower, $c_{\min} < c_{\text{lin}}$. In all cases, the estimated speed $c_{\text{est}}$ agrees with the analytical $c_{\min}$ from the dispersion relation (5.9) in the main text. All simulations use domain $x \in [-10, 10]$, with parameters $\kappa = 3 \times 10^{-4}, \mu = 0.005, \alpha = 0.01, \gamma = 0.0001$. The parameter $\gamma$ applies only to Case 5. Division rates are $\beta = 0.025$ (Cases 1, 1b), 0.0375 (Case 2), 0.0005625 (Case 3), and 0.00025 (Cases 4–5). The rates $\beta$ were selected so that all cases satisfy the same relative distance above their survival condition, with $r = 2\beta/(\mu + \alpha) = 5$ (or the corresponding expression for each model). This ensures comparable invasion regimes across models. Initial condition: $u(a, x, 0) = 0.01 \mathbb{1}_{x < -8} e^{-a}$.*

performed using L-BFGS-B with box constraints $\beta \in [10^{-3}, 2.0]$ and $\mu \in [10^{-4}, 1.0]$. Complete Python code for fitting and visualisation is available on GitHub at https://github.com/Stephanie-Abo/age-structured-invasion-dynamics.git.

**Declarations**

**Conflicts of interest.** The authors declare there are no conflicts of interest.

**Consent for publication** All the authors approved the final version of the manuscript.

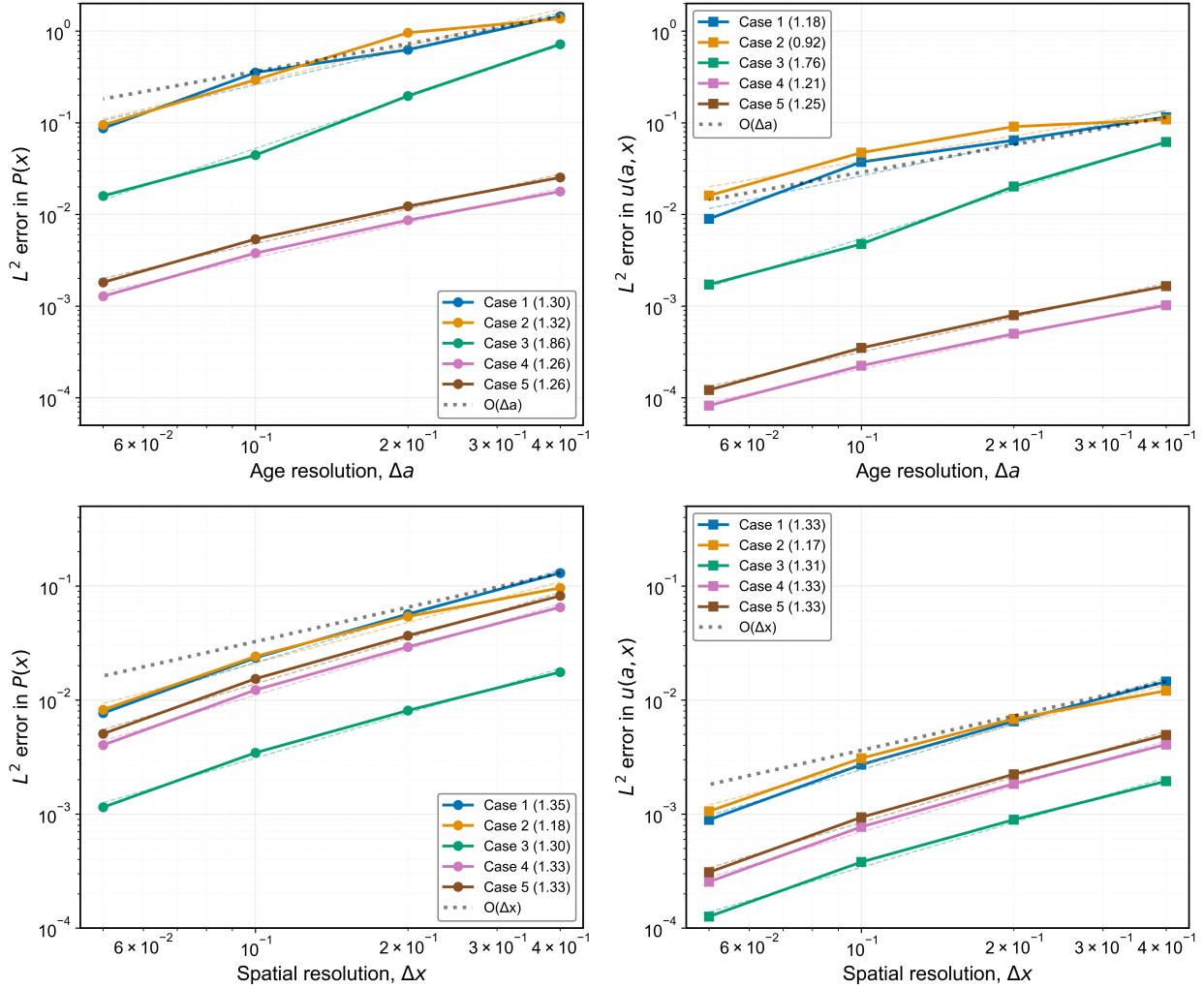

Fig. SM3. *Spatial and temporal convergence of the numerical scheme.* **Top row:** $L^2$ error in total population $P(x)$ (left) and age distribution $u(a,x)$ (right) versus age-step size $\Delta a$ (spatial resolution fixed at $\Delta x = 0.1$). **Bottom row:** $L^2$ error in $P(x)$ (left) and $u(a,x)$ (right) versus spatial-step size $\Delta x$ (age resolution fixed at $\Delta a = 0.1$). All simulations use domain $x \in [-10, 10]$, with parameters $\kappa = 3 \times 10^{-4}, \mu = 0.005, \alpha = 0.01, \gamma = 0.0001$. The parameter $\gamma$ applies only to Case 5. Division rates are $\beta = 0.025$ (Cases 1, 1b), 0.0375 (Case 2), 0.0005625 (Case 3), and 0.00025 (Cases 4–5). The rates $\beta$ were selected so that all cases satisfy the same relative distance above their survival condition, with $r = 2\beta/(\mu + \alpha) = 5$ (or the corresponding expression for each model). This ensures comparable invasion regimes across models. Initial condition: $u(a,x,0) = 0.01 \mathbb{1}_{x<-8} e^{-a}$. Errors are measured at steady wave profiles relative to finer reference grids ($\Delta a = \Delta x = 0.05$). Numbers in parentheses show convergence slopes from log–log regression; dotted lines mark first-order references $O(\Delta a)$ and $O(\Delta x)$. Age resolution gives near-first-order convergence (slopes $\approx$ 1.2–1.3, except Case 3 $\approx$ 1.8), while spatial resolution yields $\approx$ 1.3 across cases—consistent with coupling of first-order age transport with Crank–Nicolson diffusion.

**Code availability.** All code used for the simulations and figure generation is publicly available on GitHub at https://github.com/Stephanie-Abo/age-structured-invasion-dynamics.git.



\end{document}